\documentclass[10pt,colorinlistoftodos]{article}
\usepackage[utf8]{inputenc}

\usepackage[dvips]{graphics,graphicx,epsfig}
\usepackage{amsmath}

\usepackage{amssymb, latexsym, amsbsy,amsthm}
\usepackage{amsfonts}
\usepackage{booktabs}
\usepackage[T1]{fontenc}
\usepackage{amscd}
\usepackage{graphicx}
\usepackage{subcaption}
\usepackage{latexsym}
\usepackage{color}
\usepackage[colorlinks]{hyperref}
\usepackage{subeqnarray}
\usepackage{eqnarray}
\usepackage{epsfig}
\usepackage{lineno}
\usepackage{color}
\usepackage[colorlinks]{hyperref}
\usepackage{bm}
\PassOptionsToPackage{force}{filehook}
\usepackage{tikz}
\usetikzlibrary{shapes,arrows}
\usetikzlibrary{positioning}
\tikzstyle{cloud} = [draw, ellipse,fill=red!20, node distance=0.87cm,
minimum height=2em]
\tikzstyle{line} = [draw, -latex']
\usetikzlibrary{shapes.symbols,shapes.callouts,patterns}
\usetikzlibrary{calc}

\def\R{\mathcal{R}}

\newcommand{\II}{\mathbb{I}}

\newcommand{\IR}{\mathbb{R}}

\def\ba{\bm{a}}
\def\bb{\bm{b}}
\def\bc{\bm{c}}

\def\bp{\bm{p}}
\def\bu{\bm{u}}

\def\bA{\bm{A}}

\def\bE{\bm{E}}
\def\bI{\bm{I}}
\def\bL{\bm{L}}

\def\bX{\bm{X}}
\def\bY{\bm{Y}}

\def\bdelta{\bm{\delta}}
\def\bvarepsilon{\bm{\varepsilon}}
\def\bmu{\bm{\mu}}
\def\bpi{\bm{\pi}}
\def\btau{\bm{\tau}}

\def\bzero{\bm{0}}
\def\bone{\bm{1}}

\def\diag{\mathsf{diag}}

\usepackage{todonotes}
\usepackage{geometry}
\newtheorem{theorem}{Theorem}
\newtheorem{lemma}[theorem]{Lemma}

\author{
CI Hameni Nkwayep$^{1}$,
Julien Arino$^{2*}$,
PM Tchepmo Djomegni$^{1}$\\[1mm]
$^{1}$School of Mathematics and Statistical Sciences, North-West University, South Africa\\
$^{2}$Department of Mathematics, University of Manitoba, Winnipeg,
Manitoba, Canada\\
$^{*}$Corresponding author: \texttt{your.email@domain}
}
\date{\today}
\title{A model with exposure in the epidemiological sense\\ Part 1 -- Base model}
\begin{document}


\maketitle

\begin{abstract}
We explore a model of infectious disease spread that incorporates exposure to the pathogen in the classic epidemiological acceptation of the term, i.e., a contact with an infectious individual has taken place but the infection has not necessarily been acquired.
The model also includes a (discrete) age of infection structure, allowing to implicitly describe the viral load of infected individuals and in turn, to describe the probability of developing an infection as a function of the viral load of the infectious contacts.
\end{abstract}


\section{Introduction}
With the notable exception of Daniel Bernoulli \cite{bernoulli1760}, much of what is now called \emph{mathematical epidemiology} was, in early days, in the hands of physicians such as Ross \cite{ross1911prevention,ross1915some,ross1916application}, Kermack and McKendrick \cite{kermackmckendrick1927,kermackmckendrick1932,kermackmckendrick1937,kermackmckendrick1939} or MacDonald \cite{macdonald1952analysis,macdonald1953analysis,macdonald1957epidemiology}.
Then in the 1960s and 1970s, the type of models obtained when considering the spread of infectious diseases drew the attention of some of the most prominent mathematical biologists of the time.
This led to fantastic advancements in the comprehension of the dynamics of these systems.
However, this also had one particular unintended consequence: someone (probably Ken Cooke \cite{burke2024origins}) called \emph{exposed} individuals having contracted the disease but not yet infectious to others.
Thus were born SEI models and all their variations and complexifications that have since become hallmarks of mathematical epidemiology.

However, this definition of exposure flies in the face of years of ``classical'' epidemiology and public health science, where exposure \emph{does not} assume that transmission took place, just that a \emph{potentially} infecting contact took place \cite{MilwidSteriuArinoHeffernanEtAl2016}.
Some authors, including one of the authors of the present paper, have tried to get the community to use the term \emph{latently infected} and the letter L instead of \emph{exposed} and the letter E.
These efforts have been mostly unsuccessful, especially since many classical epidemiologists do understand the distinction and use the term SEIR while meaning SLIR.

So the quixotic quest to correct this wrong could probably be forgone.
However, this terminological confusion has led to some modelling mistakes or imprecisions.
The most classic error comes when modelling contact tracing leading to isolation, which is further compounded by the fact that many models confuse isolation and quarantine, with the former typically involving detected cases and the latter being indiscriminate.
Regardless of whether they correctly call it isolation or incorrectly call it quarantine, many models assume that isolation results in a fraction of the exposed getting isolated.
While this is implicitly correct for a fraction of the population since individuals exposed in the classic mathematical epidemiology sense were indeed exposed in the classic epidemiology sense, this means that only individuals having acquired the infection are isolated.
In ``real life'', unless some positive test is required before someone is asked to isolate, this is not true.
This automatically implies that if conducting a cost-benefit analysis of the value of isolation measures, the benefit of contact screening is always overestimated since the denominator of the population the measure is applied to is smaller.

Motivated by these considerations, we ponder in the present work the question of exposure.
Specifically, we formulate and analyse a model in which exposure is in the classical epidemiological sense.
A thorough mathematical analysis is conducted in appendix.
Then we consider numerically properties of the model, focusing in particular on differences with a corresponding classic model where ``E means L''.

\section{The model}
\label{sec:math-model}

Before proceeding further, let us set some notation.
By default, vectors are assumed to be column vectors; however, we only indicate directionality if omitting it leads to confusion.
If $\ba,\bb$ are two vectors, we denote $(\ba,\bb)$ the vector $(\ba^T,\bb^T)$.
Given a $k$-vector $\ba$, $\diag(\ba)$ is the $k\times k$ diagonal matrix with elements of $\ba$ on its main diagonal.
We denote $\bone_k=(1,\ldots,1)$ the $k$-vector of all ones, $\mathbb{I}_k$ the $k\times k$ identity matrix, $\circ$ the Hadamard product and $\langle\ba,\bb\rangle$ the dot product of $\ba$ and $\bb$.
Regarding collections of zeros, we write $\bzero_{m\times n}$ the $m\times n$ zero matrix and $\bzero_{k}$ the vector with $k$ entries, bearing in mind the previous remark about orientation.
For state variables, we typically omit time dependence; vectors of state variables are denoted using bold letters.

\subsection{State variables}
\label{sec:model-state-variables}

We consider a population divided into compartments based on the epidemiological status of individuals.
Susceptible individuals $S$ are susceptible to the disease.
Corresponding to the susceptible compartment is a group of \emph{exposed} compartments $E$ containing individuals having been exposed to the pathogen when they were in the susceptible compartment; see below.
If exposure leads to infection, then exposed individuals progress into one of the two infected and infectious compartments $I$ or $A$ depending on whether they become \emph{symptomatic} or \emph{asymptomatic}, respectively.
Otherwise, exposed individuals revert to being susceptible.
Infected compartments are further subdivided.
A precise explanation is given in Section~\ref{sec:model-assumptions}; for now, it suffices to say that both $I$ and $A$ are divided into $n$ compartments $I_1,\ldots,I_n$ and $A_1,\ldots,A_n$.
We often write them as vectors $\bI=(I_1,\ldots,I_n)$ and $\bA=(A_1,\ldots,A_n)$.
Finally, if they do not die as a result of the infection, infected individuals progress to the (single) \emph{recovered} compartment $R$ where they enjoy temporary disease-induced immunity to reinfection.

To accurately capture the outcome of an exposure, which depends on the characteristics of the infectious contact, we track the specific source of exposure. 
We denote $E_k^X$ the compartment containing individuals who were exposed via contact with an infectious individual in compartment $X_k$, where $X \in \{I, A\}$ and $k=1,\ldots,n$.
It is convenient in places to use a vector notation for these $2n$ compartments, so we define
\[
    \bE = \left(E_1^I, \ldots, E_n^I, E_1^A, \ldots, E_n^A\right).
\]
The total exposed population is $E = \langle \bone_{2n},\bE \rangle$. 
It is also useful later to refer to infected compartments only, so we define
\[
    \bY = \left(I_1,\ldots,I_n,A_1,\ldots,A_n\right)=(\bI,\bA).
\]
The full state variable vector is
\[
    \bX =(S,\bE,\bI,\bA,R)= \left(S, \bE, \bY, R\right).
\]
The total population is given by
\begin{equation}
\label{eq:total}
N=S+ E +\sum_{i=1}^nI_i+\sum_{i=1}^nA_i+R = \langle\bone_{4n+2},\bX\rangle.
\end{equation}

\subsection{Model assumptions}
\label{sec:model-assumptions}

\paragraph{Basic demography.} Recruitment into the population is to the susceptible compartment $S$; there is no vertical transmission of the disease.
Recruitment occurs at the constant rate $b$.
All compartments are subject to natural death, which occurs at the \emph{per capita} rate $d$.

\paragraph{Times of sojourn in infected compartments are ``essentially'' Erlang-distributed.}
As mentioned in Section~\ref{sec:model-state-variables}, infected compartments are further subdivided into $n$ compartments.
The reason for this subdivision is so that the time spent in these compartments is Erlang distributed; see, e.g., \cite[Section 3.1]{Arino2020a}.
We assume that the rates of progression through and out of symptomatic and asymptomatic compartments are $\gamma_I$ and $\gamma_A$, respectively.
We also assume that all infectious individuals are subject to natural death at the \emph{per capita} rate $d$.
This means that individuals in $\bI$ compartments are subject to competing exponentially distributed risks with parameters $\gamma_I$ and $d$ and, as a consequence, the time of sojourn in each $I_i$ compartment is exponentially distributed with parameter $\gamma_I+d$.
Likewise, times of sojourn in compartments $A_i$ are exponentially distributed with parameters $\gamma_A+d$.
As a consequence, the times of sojourn in the chains of $\bI$ and $\bA$ compartments, in the absence of other flows, are Erlang distributed with shape parameters $n$ and scale parameters $\gamma_I+d$ and $\gamma_A+d$, respectively.
This implies that their means are $n/\gamma_I$ and $n/\gamma_A$ and their variances are $n/(\gamma_I+d)^2$ and $n/(\gamma_A+d)^2$, respectively.

\begin{figure}[htbp]
    \centering
    \def\hhskip{*2.5}
    \def\vvskip{*4}
    \begin{tikzpicture}[scale=0.65,
    every node/.style={transform shape},
    auto,
    box/.style={minimum size=1.5cm,
    draw=black, very thick},
    draw, circle]
    
    \draw [black,rounded corners, thick, fill=red!50] (0.15\hhskip,0.75\vvskip) -- (5.5\hhskip,0.75\vvskip) -- (5.5\hhskip,0.25\vvskip) -- (0.15\hhskip,0.25\vvskip) -- cycle;
    \draw [black,rounded corners, thick, fill=red!50] (0.15\hhskip,-0.25\vvskip) -- (5.5\hhskip,-0.25\vvskip) -- (5.5\hhskip,-0.75\vvskip) -- (0.15\hhskip,-0.75\vvskip) -- cycle;
    \draw [black,rounded corners, thick, fill=red!20] (-1.8\hhskip,1.12\vvskip) -- (-0.6\hhskip,1.12\vvskip) -- (-0.6\hhskip,-1.12\vvskip) -- (-1.8\hhskip,-1.12\vvskip) -- cycle;

    \node [box, fill=gray!20] at (-1.2\hhskip,0.9\vvskip) (EI1) {$E_1^I$};
    \node [box, fill=gray!20] at (-1.2\hhskip,0.22\vvskip) (EIn) {$E_n^I$};
    \draw [thick, dotted] (EI1) -- (EIn);

    \node [box, fill=gray!20] at (-1.2\hhskip,-0.22\vvskip) (EA1) {$E_1^A$};
    \node [box, fill=gray!20] at (-1.2\hhskip,-0.9\vvskip) (EAn) {$E_n^A$};
    \draw [thick, dotted] (EA1) -- (EAn);

    \node [box, fill=green!40] at (6.5\hhskip,0\vvskip) (R) {$R$};
    
    \node [box, fill=gray!20] at (0.65\hhskip,0.5\vvskip) (I1) {$I_{1}$};
    \node [box, fill=gray!20] at (2.75\hhskip,0.5\vvskip) (Ik) {$I_{k}$};
    \node [box, fill=gray!20] at (5\hhskip,0.5\vvskip) (In) {$I_{n}$};
    \coordinate (I2) at (1.3\hhskip,0.5\vvskip);
    \coordinate (I3) at (1.4\hhskip,0.5\vvskip);
    \coordinate (Ikm2) at (1.9\hhskip,0.5\vvskip);
    \coordinate (Ikm1) at (2\hhskip,0.5\vvskip);
    \coordinate (Ikp1) at (3.55\hhskip,0.5\vvskip);
    \coordinate (Ikp2) at (3.65\hhskip,0.5\vvskip);
    \coordinate (Inm2) at (4.15\hhskip,0.5\vvskip);
    \coordinate (Inm1) at (4.25\hhskip,0.5\vvskip);
    
    \node [box, fill=gray!20] at (0.65\hhskip,-0.5\vvskip) (A1) {$A_{1}$};
    \node [box, fill=gray!20] at (2.75\hhskip,-0.5\vvskip) (Ak) {$A_{k}$};
    \node [box, fill=gray!20] at (5\hhskip,-0.5\vvskip) (An) {$A_{n}$};
    \coordinate (A2) at (1.3\hhskip,-0.5\vvskip);
    \coordinate (A3) at (1.3\hhskip,-0.5\vvskip);
    \coordinate (Akm2) at (1.9\hhskip,-0.5\vvskip);
    \coordinate (Akm1) at (2\hhskip,-0.5\vvskip);
    \coordinate (Akp1) at (3.55\hhskip,-0.5\vvskip);
    \coordinate (Akp2) at (3.65\hhskip,-0.5\vvskip);
    \coordinate (Anm2) at (4.15\hhskip,-0.5\vvskip);
    \coordinate (Anm1) at (4.25\hhskip,-0.5\vvskip);
    
    \coordinate[above=0.4\vvskip of I1] (deathI1);
    \coordinate[above=0.4\vvskip of Ik] (deathIk);
    \coordinate[above=0.4\vvskip of In] (deathIn);
    \coordinate[below=0.4\vvskip of A1] (deathA1);
    \coordinate[below=0.4\vvskip of Ak] (deathAk);
    \coordinate[below=0.4\vvskip of An] (deathAn);
    
    \path [line, thick] (EI1) to (I1);
    \path [line, thick] (EI1) to (A1);
    \path [line, thick] (EIn) to (I1);
    \path [line, thick] (EIn) to (A1);
    \path [line, thick] (EA1) to (I1);
    \path [line, thick] (EA1) to (A1);
    \path [line, thick] (EAn) to (I1);
    \path [line, thick] (EAn) to (A1);
    
    \path [line, thick] (I1) to node [midway, above] (TextNode) {$\gamma_I I_1$} (I2);
    \path [line, thick, dashed] (I3) to (Ikm2);
    \path [line, thick] (Ikm1) to node [midway, above] (TextNode) {$\gamma_II_{k-1}$} (Ik);
    \path [line, thick] (Ik) to node [midway, above] (TextNode) {$\gamma_II_{k}$} (Ikp1);
    \path [line, thick,dashed] (Ikp2) to (Inm2);
    \path [line, thick] (Inm1) to node [midway, above] (TextNode) {$\gamma_I I_{n-1}$} (In);
    \path [line, thick] (In) to node [near end, above] (TextNode) {$\gamma_I I_{n}$} (R);
    
    \path [line, thick] (A1) to node [midway, above] (TextNode) {$\gamma_A A_1$} (A2);
    \path [line, thick, dashed] (A3) to (Akm2);
    \path [line, thick] (Akm1) to node [midway, above] (TextNode) {$\gamma_AA_{k-1}$} (Ak);
    \path [line, thick] (Ak) to node [midway, above] (TextNode) {$\gamma_AA_{k}$} (Akp1);
    \path [line, thick,dashed] (Akp2) to (Anm2);
    \path [line, thick] (Anm1) to node [midway, above] (TextNode) {$\gamma_AA_{n}$} (An);
    \path [line, thick] (An) to node [near end, below] (TextNode) {$\gamma_A A_{n}$} (R);
    
    \path [line, thick] (A1) to node [midway, right] (TextNode) {$\tau_1A_1$} (I1);
    \path [line, thick] (Ak) to node [midway, right] (TextNode) {$\tau_kA_k$} (Ik);
    \path [line, thick] (An) to node [midway, right] (TextNode) {$\tau_nA_n$} (In);
    
    \path [line, thick] (I1) to node [midway, right] (TextNode) {$(d+\mu_1)I_1$} (deathI1);
    \path [line, thick] (Ik) to node [midway, right] (TextNode) {$(d+\mu_k)I_k$} (deathIk);
    \path [line, thick] (In) to node [midway, right] (TextNode) {$(d+\mu_n)I_n$} (deathIn);
    \path [line, thick] (A1) to node [midway, right] (TextNode) {$dA_1$} (deathA1);
    \path [line, thick] (Ak) to node [midway, right] (TextNode) {$dA_k$} (deathAk);
    \path [line, thick] (An) to node [midway, right] (TextNode) {$dA_n$} (deathAn);
    \end{tikzpicture}
    \caption{Flows into, within and out of infectious compartments $I$ and $A$.
    Only rates related to flows within or out from infectious compartments are shown.}
    \label{fig:flow-within-infectious}
\end{figure}
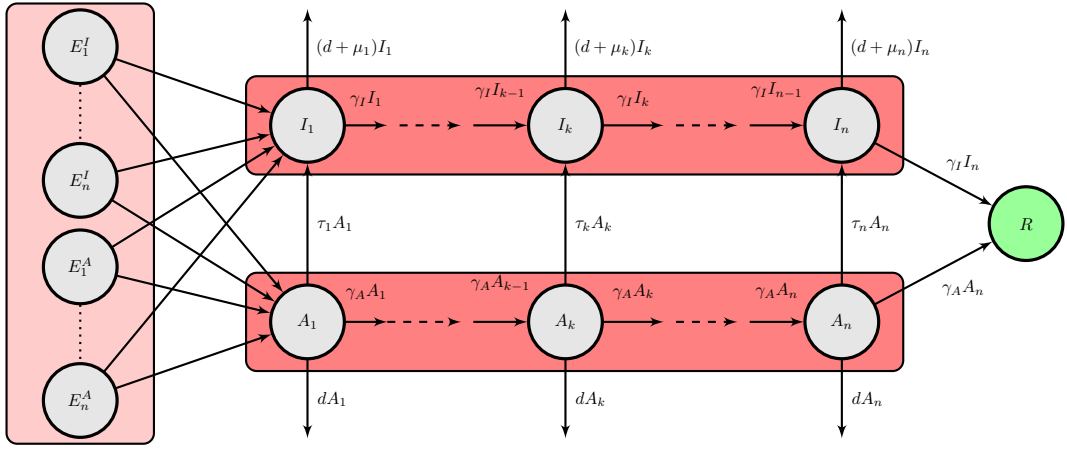

In theory, we could suppose that the numbers of compartments differ between $\bI$ and $\bA$, but this leads to unnecessary complications.
Indeed, since we allow individuals to move from $\bA$ to $\bI$, having discrepant compartment numbers would require to connect the relevant $A_i$ to the corresponding $I_j$.
Since $\gamma_I$ and $\gamma_A$ are not necessarily equal, this means we are connecting compartments in which individuals have spent potentially different amounts of time.
However, given the imprecision involved in describing sojourn times using ordinary differential equations, we feel this is an acceptable simplification.
As a consequence, we use the position (index) in the chain as a proxy for viral load of individuals in the corresponding compartment.

Adding rates of movement from $\bA$ to $\bI$ and of disease-induced death in $\bI$ that depend on $i=1,\ldots,n$ means that the $I_i$ and $A_i$ have sojourn times that are no longer identically distributed, breaking the \emph{strict} Erlang distributed nature of the overall process.
However, sojourn times in $\bI$ and $\bA$ retain a \emph{generalised Erlang} distribution; see Section~\ref{sec:viral-load-mapping}.
By abuse of language, we still call $\bI$ and $\bA$ Erlang chains.

\paragraph{Contacts leading to exposure to the disease.}
In this model, we use \emph{exposure} in the classic epidemiological sense, i.e., the word \emph{exposed} does not (incorrectly \cite{burke2024origins,MilwidSteriuArinoHeffernanEtAl2016}) represent incubation with the disease.
Exposure occurs when a contact takes place between a susceptible and an infectious individual.
Upon exposure, the susceptible individual receives a certain dose of pathogen and the outcome of exposure depends on both the dose received and the capacity of the host to deal with this dose of pathogen.
To accurately capture the outcome of an exposure, which is dependent on the viral dose and characteristics of the infectious contact, we structure the exposed population to retain information about the source of infection. 
Therefore, there are $2n$ sub-compartments in $\bE$, with $n$ corresponding to the different compartments in $\bI$ and $n$ corresponding to those in $\bA$.

Individuals enter compartment $E_k^X$ ($X \in \{I, A\}$) strictly through contact with an infectious individual in the corresponding state $X_k$. 
Letting $c_k^X$ be the specific contact rate, the flow of newly exposed individuals into $E_k^X$ is $c_k^X X_k {S}/{N}$, i.e., we assume proportional exposure.
We call this intensity \emph{force of exposure} to distinguish it from the usual force of infection, since the latter implies that an infection does occur whereas we assume no such thing.

The \emph{total} force of exposure is the sum over all infectious compartments,
\begin{equation}
\label{eq:force-of-exposure}
\lambda_\text{exp}=\dfrac{\sum_{k=1}^n c^I_kI_k
	+\sum_{k=1}^n c^A_kA_k}{N}= \frac{\left\langle\bc,\bY\right\rangle}{\langle\bone_{4n+2},\bX\rangle},
\end{equation}
where $\bc = \left(c_1^I,\ldots,c_n^I,c_1^A,\ldots,c_n^A\right)$.

Contact rates for asymptomatically infectious individuals would not vary much as $i$ varies, because such individuals are often unaware of their status.
On the other hand, the contact rates of symptomatically infectious might vary, because sickness makes them less likely to seek contacts or some social distancing measure is imposed onto them.

\paragraph{Exposed compartments to infectious or susceptible compartments.}
Once exposed, individuals spend on average $1/\varepsilon_k^X$ time units in their respective $E_k^X$ compartment, $X\in\{I,A\}$ and $k=1,\ldots,n$. 
After this ``delay'', they either develop the infection or revert to being susceptible.
In vector form, we denote $\bvarepsilon=(\varepsilon_1^I,\ldots,\varepsilon_n^I,\varepsilon_1^A,\ldots,\varepsilon_n^A)$.

Let $\delta_k^X \in [0,1]$ be the proportion of individuals leaving $E_k^X$ who develop the infection. 
The complementary fraction $(1-\delta_k^X)$ return to the susceptible pool $S$. 
Thus, progression to infection is dependent on the viral dose received during the contact and the capacity for a host to rapidly clear pathogens \cite{PALUDAN2024115070}.
We denote $\bdelta=(\delta_1^I,\ldots,\delta_n^I,\delta_1^A,\ldots,\delta_n^A)$ the corresponding vector.

\paragraph{Determining the proportion of symptomatic infections.}
Among exposed individuals who progress towards infection, the clinical outcome depends on the source of the exposure. 
We denote $\pi_k^X \in [0,1]$ the proportion of newly infected individuals originating from $E_k^X$ who develop a symptomatic infection. 
Consequently, they enter $I_1$ at a rate $\pi_k^X \varepsilon_k^X \delta_k^X E_k^X$, while those remaining asymptomatic enter $A_1$ at a rate $(1-\pi_k^X)\varepsilon_k^X \delta_k^X E_k^X$.

Denoting $\bpi=(\pi_1^I,\ldots,\pi_n^I,\pi_1^A,\ldots,\pi_n^A)$, the vectors describing the rate of progression of exposed individuals into the symptomatic and asymptomatic compartments are $\bp_I=\bpi\circ\bvarepsilon\circ\bdelta$ and $\bp_A=(\bone_{2n}-\bpi)\circ\bvarepsilon\circ\bdelta$, respectively.

\paragraph{Transition from asymptomatic to symptomatic infections.}
The transition from $A$ to the symptomatic class $I$ reflects the moment when viral replication is sufficient to trigger a clinical immune response and the onset of symptoms. 
This progression is not simultaneous for all individuals: some may already be transmitting the virus before detection by PCR or antigen testing, including at the onset of symptoms \cite{He2020,Kucirka2020}. 
The model therefore considers two pathways out of $\bA$: $\bA$ to $\bI$ for presymptomatic individuals who develop symptoms and $A_n$ to $R$ for asymptomatic or presymptomatic individuals who recover without ever developing symptoms. 
This distinction allows to account for the biological dynamics of viral load, the variability of contagiousness and the existence of early false negatives, while remaining consistent with clinical and epidemiological data \cite{Byambasuren2020,Sayampanathan2021}.
The rate at which individuals with an asymptomatic infection in compartment $A_k$ progress to the corresponding compartment $I_k$ is denoted $\tau_k$.
The corresponding vector is denoted $\btau=(\tau_1,\ldots,\tau_n)\in\IR^n$. 

\paragraph{Disease-induced death in symptomatic infections.}
Symptomatic infections are more severe than asymptomatic ones.
For this reason, we assume that some level of disease-induced death takes place for individuals in $\bI$ compartments.
This rate increases with the viral load \cite{pujadas2020sars,towner2004rapid}.
To model this, we let disease-induced death in compartment $I_i$, $i=1,\ldots,n$, occur at the \emph{per capita} rate $\mu_i$, with the vector of disease-induced death rates denoted $\bmu=(\mu_1,\ldots,\mu_n)$.

\paragraph{Temporary immunity.}
If they make it to the end of the $\bI$ or $\bA$ Erlang chains without dying from the disease or natural causes, individuals progress to the R compartment, where they enjoy temporary immunity to the disease.
They lose this temporary immunity after an average $1/\nu$ time units, at which point they move back into the susceptible compartment.

\subsection{The model}
\label{sec:modelling-model}

\begin{figure}[htbp]
\centering
\def\hhskip{*4.2} 
\def\vvskip{*2} 
\begin{tikzpicture}[scale=0.8, 
        every node/.style={transform shape},
        auto,
        param/.style={draw, fill=white, rectangle, inner sep=2pt, align=center},
        cloud/.style={minimum width=30pt, draw, circle, fill=gray!20},
        erlang/.style={minimum width=60pt, minimum height=30pt, draw, rounded corners},
        distrib/.style={minimum width=30pt, minimum height=60pt, draw, rounded corners},
        line/.style={draw, -latex'}]    
    \node[cloud, fill=blue!30] (S) at (0\hhskip,0\vvskip) {$S$};
    \node[distrib, fill=red!10] (E) at (1\hhskip,0\vvskip) {$\bE$};
    \node[erlang, fill=red!80] (I) at (2\hhskip,1\vvskip){$\bI$};
    \node[erlang, fill=red!80] (A) at (2\hhskip,-1\vvskip){$\bA$};
    \node[cloud, fill=green!40] (R) at (3\hhskip,0\vvskip){$R$};
    \node [cloud, left=0.3\hhskip of S, draw=none, fill=none] (birthS) {};
    \node [cloud, shift={(270:0.8\hhskip)}, draw=none, fill=none] (deathS) at (S) {};
    \node [cloud, shift={(270:0.8\hhskip)}, draw=none, fill=none] (deathE) at (E) {};
    \node [cloud, shift={(0:0.9\hhskip)}, draw=none, fill=none] (deathI) at (I) {};
    \node [cloud, shift={(0:0.6\hhskip)}, draw=none, fill=none] (deathR) at (R) {};
    \node [cloud, shift={(0:0.9\hhskip)}, draw=none, fill=none] (deathA) at (A) {};
    \draw[->, thick] (birthS) to node[above] {$b$} (S);
    \draw[->, thick] (S) to node[above,sloped] {$dS$} (deathS);
    \draw[->, thick] (E) to node[above,sloped] {$d\bE$} (deathE);
    \draw[->, thick] (I) to node[above,sloped] {$(d\bone_n+\bmu)\circ\bI$} (deathI);
    \draw[->, thick] (A) to node[above,sloped] {$d\bA$} (deathA);
    \draw[->, thick] (R) to node[above] {$dR$} (deathR);
    \draw[->, thick, bend right=10] (S) to node [below,sloped] {$\lambda_{\text{exp}} S$} (E);
    \draw[->, thick, bend right=10] (E) to node [above,sloped] {$(\bone_n-\bdelta)\circ\bvarepsilon\circ\bE$} (S);
    \draw[->, thick] (E) to node [above,sloped] {$\bpi\circ\bdelta\circ\bvarepsilon\circ\bE$} (I);
    \draw[->, thick] (E) to node [above,sloped,midway] {$(\bone_n-\bpi)\circ\bdelta\circ\bvarepsilon\circ\bE$} (A);
    \draw[->, thick] (I) to node [above,sloped] {$\gamma_II_n$} (R);
    \draw[->, thick] (A) to node [above,sloped] {$\gamma_AA_n$} (R);
    \draw[->, thick] (A) to node [sloped,above] {$\btau\bA$} (I);
    \draw[->, thick] (R) to (3\hhskip,1.35\vvskip) to node [above] {$\nu R$} (0\hhskip,1.35\vvskip) to (S);
\end{tikzpicture}
\caption{Flowchart of the model. 
Horizontal rectangular nodes are ``Erlang chains'', i.e., they comprise a number of similar compartments; see Figure~\ref{fig:flow-within-infectious} for details.
The vertical rectangular node is the discrete distribution $\bE=(E_1^I,\ldots,E_n^I,E_1^A,\ldots,E_n^A)$.}
\label{fig:flow-diagram-no-vacc}
\end{figure}
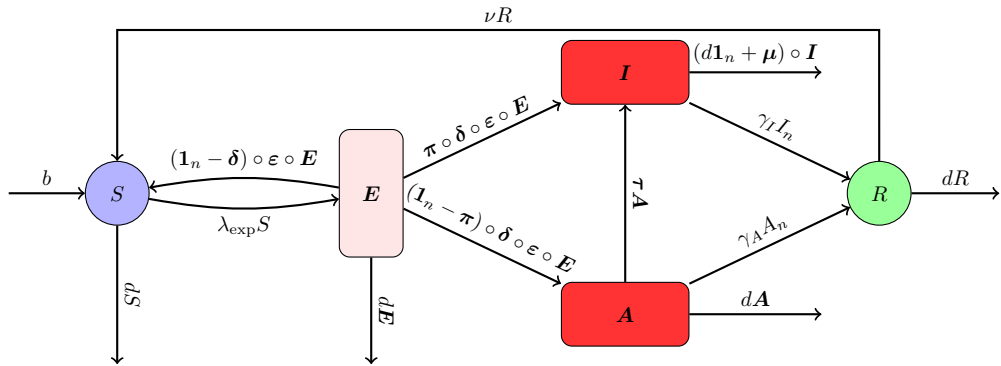

Putting together all these considerations, we obtain the flowchart in Figure~\ref{fig:flow-diagram-no-vacc} and the transmission model takes the form of the following deterministic system of nonlinear ordinary differential equations,
\begin{subequations}
\label{sys:SEIARS}
\begin{align}
\dot S &= b + \sum_{k=1}^n \left[ (1-\delta_k^I)\varepsilon_k^I E_k^I + (1-\delta_k^A)\varepsilon_k^A E_k^A \right] +\nu R -\left(\lambda_{\text{exp}}+d\right)S, \label{sys:SEIARS-dS} \\
\dot E_k^I &= \frac{c_k^I I_k}{N} S - (\varepsilon_k^I+d)E_k^I, 
&&  k=1,\dots,n, \label{sys:SEIARS-dEI} \\
\dot E_k^A &= \frac{c_k^A A_k}{N} S - (\varepsilon_k^A+d)E_k^A, 
&& k=1,\dots,n, \label{sys:SEIARS-dEA} \\
\dot I_1 &= \sum_{k=1}^n \left[ \pi_k^I \delta_k^I\varepsilon_k^I E_k^I + \pi_k^A \delta_k^A\varepsilon_k^A E_k^A \right] + \tau_1 A_1 - (\gamma_I+\mu_1+d)I_1, \label{sys:SEIARS-dI1} \\
\dot I_k &= \gamma_I I_{k-1} + \tau_k A_k - (\gamma_I+\mu_k+d)I_k, && k=2,\dots,n,\label{sys:SEIARS-dIk} \\
\dot A_1 &= \sum_{k=1}^n \left[ (1-\pi_k^I) \delta_k^I \varepsilon_k^I E_k^I + (1-\pi_k^A) \delta_k^A\varepsilon_k^A E_k^A \right] - (\gamma_A+\tau_1+d)A_1, \label{sys:SEIARS-dA1} \\
\dot A_k &= \gamma_A A_{k-1} - (\gamma_A+\tau_k+d)A_k, && k=2,\dots,n, \label{sys:SEIARS-dAk} \\
\dot R &= \gamma_I I_n + \gamma_A A_n - (\nu+d)R. \label{sys:SEIARS-dR}
\end{align}
\end{subequations}

The summary and definitions of variables and parameters of model are given in Tables~\ref{tab:variables} and \ref{tab:parameters}, respectively.

\begin{table}[htbp]
\centering
\begin{tabular}{cl}
\toprule
Variables & Denominations \\
\midrule
$S$ & Susceptible individuals \\
$\bE^I=(E_1^I,\ldots,E_n^I)$ & Individuals exposed to symptomatically infectious individuals \\
$\bE^A=(E_1^A,\ldots,E_n^A)$ & Individuals exposed to asymptomatically infectious individuals \\
$\bE=(\bE^I,\bE^A)$ & Individuals exposed to the pathogen \\
$\bI=(I_1,\ldots,I_n)$ & Symptomatically infectious individuals \\
$\bA=(A_1,\ldots,A_n)$ & Asymptomatically infectious individuals \\
$R$ & Temporarily immune individuals \\
\bottomrule
\end{tabular}
\caption{State variables in \eqref{sys:SEIARS}.}
\label{tab:variables}
\end{table}

\begin{table}
\centering
\begin{tabular}{ll}
\toprule
Parameters & Definitions \\
\midrule
$b$ & Birth rate \\
$d$ & Natural mortality rate \\
$1/\varepsilon_k^X$ & Average time in exposed compartment \\
$\pi_k^X$ & Proportion of exposed susceptible indiv. presenting symptoms \\
$\gamma_I$ & Recovery rate of symptomatic infectious indiv. \\
$\gamma_A$ & Recovery rate of asymptomatic infectious indiv. \\
$\delta_k^X$ & Proportion of exposed susceptible indiv. becoming infectious \\
$\mu_k$ & Disease-induced death rate \\
$\tau_k$ & Rate of asymptomatic indiv. developing symptoms \\
$1/\nu$ & Average duration of disease-induced immunity \\
\bottomrule
\end{tabular}
\caption{Parameters used in \eqref{sys:SEIARS}. 
When present, $X\in\{I,A\}$ and $k=1,\ldots,n$.
Fractions are dimensionless and take values in $[0,1]$. 
All other parameters have units \emph{per unit time}, except for ``birth'', which has units \emph{individual per unit time}.}
\label{tab:parameters}
\end{table}

\section{Mathematical analysis}
\label{sec:mathematical-analysis}

The following results, whose proofs are found in Appendix \ref{app:proofs}, summarise what we know of the behaviour of \eqref{sys:SEIARS}.
The model has the basic properties expected of epidemiological models, in that it is easy to show that it verifies positivity and boundedness of solutions and that its solutions exist and are unique.
This is shown in Appendix~\ref{app:proofs-well-posedness}.
\begin{lemma}
\label{th:positivity}
System \eqref{sys:SEIARS} is a dynamical system in the biologically feasible domain
\begin{equation}
\label{om}
    \Omega:=\left\{ (S,\bE,\bI,\bA,R)\in\IR^{4n+2},\, N\leq N(0)+\dfrac{b}{d}\right\}.
\end{equation}
For each initial condition $\bX(0)\in \Omega$, \eqref{sys:SEIARS} admits a unique maximal solution.
\end{lemma}

The disease-free equilibrium (DFE) of \eqref{sys:SEIARS} is obtained by assuming that $\bI=\bA=\bzero_n$.
Denote $\bX^0$ that point.
A simple computation shows that at $\bX^0$, all state variables are zero except for $S$.
We write this as
\[
\bX^0:=\left(S^0,\bzero_{4n+1}\right),
\]
where $S^0=b/d$. 
Note that at the disease-free equilibrium, the total population is $N(0) = S^0$ and therefore $S^0/N(0) = 1$.
The behaviour of \eqref{sys:SEIARS} is classic and is governed by the following result.

\begin{theorem}
\label{th:behaviour-fct-R_0}
Define the basic reproduction number $\R_0$ of \eqref{sys:SEIARS} as
\begin{equation}
\R_0 = 
\rho \left( \mathsf{diag}(\bc) 
\begin{pmatrix}
W_I^{-1} & W_I^{-1}\diag(\btau) W_A^{-1} \\
\bzero_{n\times n} & W_A^{-1}\end{pmatrix}
P (\diag(\bvarepsilon) + d\mathbb{I}_{2n})^{-1} \right),
\label{eq:R_0}
\end{equation}
where $\rho(\cdot)$ denotes the spectral radius and $W_I^{-1}$, $W_A^{-1}$ and $P$ take the forms in Appendix~\ref{app:proofs-R_0-LAS}.
Then the disease-free equilibrium $\bX^0$ of \eqref{sys:SEIARS} is globally asymptotically stable when $\R_0\leq 1$ and unstable if $\R_0 > 1$.
In the latter case, a unique endemic equilibrium $\bX^\star$ becomes biologically relevant.
\end{theorem}

The proof of this result involves several different components and is deferred to Appendix~\ref{app:proofs}.
Computation of $\R_0$ and considerations on the local stability or instability of the disease-free equilibrium is carried out in Appendix~\ref{app:proofs-R_0-LAS}.
Of note here is that in the $\R_0$ computation, we define the infected state vector as $(\bE,\bI,\bA)$, thereby explicitly including the exposed compartments $\bE$ among infected states.
This may not seem obvious at first glance, but consider the following: because $\bm{\delta}$ is defined as a vector of constants, if we did not include $\bE$ among the infected states, there could be an exogenous flow from $\bE$ to $I_1$ and $A_1$ even in absence of infectious individuals, which would not be realistic.

Proof that the disease-free equilibrium is globally asymptotically stable when $\R_0\leq 1$ uses the method of \cite{KamgangSallet2008} and is found in Appendix~\ref{app:proofs-GAS-DFE}.
Existence of the unique endemic equilibrium $\bX^\star$ when $\R_0>1$ is carried out in Appendix~\ref{app:existence-uniqueness-EEP}.

\section{Computational considerations}
\label{sec:computational-work}

\subsection{Parameters}
\label{sec:computational-parameters}

Table~\ref{tab:parameters-values} details the numerical values and ranges used for the baseline simulations throughout this section, corresponding to the variables defined in Table~\ref{tab:parameters}. Note that the transmission parameters $c_i^X$ are specific to the $X_i$ sub-compartments and reflect a baseline ratio.
To ensure consistency, all simulations are standardised to a baseline population size of $100,000$ individuals ($N = b/d$) matching the ``\emph{per} 100,000 individuals'' frequently used in epidemiological work.

\begin{table}[htbp]
\centering
\begin{tabular}{llp{8.5cm}}
\toprule
Parameters & Ranges & Notes \\
\midrule
$b$ & $4.0$ & Recruitment rate ($N = b/d$) \\
$d$ & $0.00004$ & Natural mortality rate (life expectancy $\sim 70$ years) \\
$1/\varepsilon_i^I$ & $[1, 5]$ & Average days in exposed compartments (symptomatic source) \\
$1/\varepsilon_i^A$ & $[1, 5]$ & Average days in exposed compartments (asymptomatic source) \\
$\pi_i^I$ & $[0.4, 0.6]$ & Proportion of exposed (symptomatic source) individuals presenting symptoms \\
$\pi_i^A$ & $[0.4, 0.6]$ & Proportion of exposed (asymptomatic source) individuals presenting symptoms \\
$1/\gamma_I$ & $10$ & Average days in symptomatic infectious stages \\
$1/\gamma_A$ & $10$ & Average days in asymptomatic infectious stages \\
$\delta_i^I$ & $[0.4, 0.8]$ & Proportion of exposed to symptomatic individuals becoming infectious \\
$\delta_i^A$ & $[0.2, 0.4]$ & Proportion of exposed to asymptomatic individuals becoming infectious \\
$\mu_i$ & $[0.0005, 0.002]$ & Disease-induced death rate per stage \\
$\tau_i$ & $[0.005, 0.02]$ & Rate of asymptomatic individuals developing symptoms per stage \\
$\nu$ & $0.0$ & Rate of loss of disease-induced immunity \\
$c_i^I$ & $[0.002, 0.01]$ & Transmission coefficients for symptomatic contacts \\
$c_i^A$ & $[0.001, 0.005]$ & Transmission coefficients for asymptomatic contacts \\
\bottomrule
\end{tabular}
\caption{Baseline parameter values used in the numerical simulations for a population of $N=100,000$. All time-dependent parameters have units \emph{per day}. Fractions are dimensionless.}
\label{tab:parameters-values}
\end{table}

\subsection{Sensitivity analysis}
\label{sec:sensitivity-analysis}

To understand how parameters influence various quantities, we first proceed to a partial rank correlation coefficient analysis of the response of several quantities in the model.
We sample 1 million values of parameters in the ranges shown in Table~\ref{tab:parameters-values} using Latin hypercube sampling (LHS) and compute the partial rank correlation coefficient of the response to changes.

\begin{figure}[htbp]
        \includegraphics[width=\linewidth]{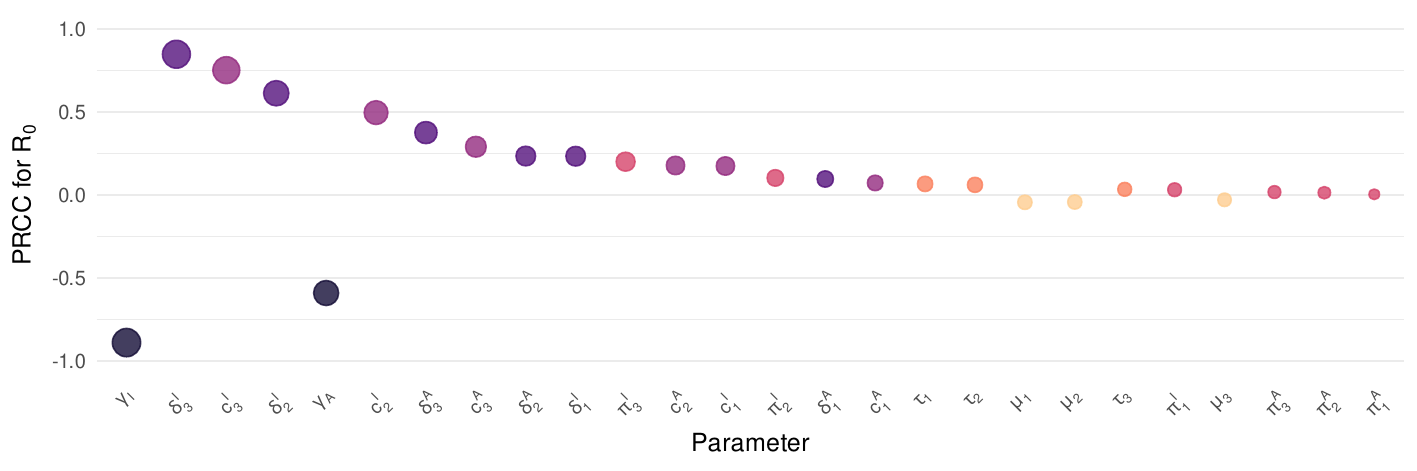}
    \caption{PRCC of $\R_0$ to model parameters.}
    \label{fig:PRCC-R0}
\end{figure}
In Figure~\ref{fig:PRCC-R0}, we show the response of $\R_0$ on parameter changes.
First, note that we are not showing the $\varepsilon_k^X$: the largest PRCC of any of them is 0.0009, meaning that their influence on $\R_0$ is negligible.
As a consequence, to simplify the numerics, we consider from now on a single value $\varepsilon$.
The second interesting observation is that in ``parameter groups'' $c_k^X$, $\delta_k^X$, $\pi_k^X$ and the not shown $\varepsilon_k^X$, for a given $X\in\{I,A\}$, the largest magnitude PRCC is always for $k=3$, then $k=2$, then $k=1$.
The situation is the opposite for $\mu_k$ and $\tau_k$, where the rates in the early compartments in the Erlang chains influence $\R_0$ more than later ones.

\subsection{Comparison with models with ``exposed'' compartment}
\label{sec:comparison-SEIRS}
We introduced \eqref{sys:SEIARS} pondering if it can ``correct some of the wrongs'' of classic SEIRS models, so one obvious question is \emph{how different are the conclusions of the two types of models}?
We investigate the question mostly computationally. 
We consider a ``classic'' SLIARS model, which is detailed in Appendix~\ref{app:SLIARS-model}.
In order to ``compare comparable models'', we do retain the Erlang structure of $\bI$ and $\bA$.
The model takes the form \eqref{sys:SLIARS}.
Note that we ``re-use'' existing parameters: instead of being interpreted as the mean duration of the exposure period,  in the SLIARS model $1/\varepsilon$ is the comparable mean duration of the latent period.
Likewise, the proportion $\bdelta$ of individuals developing an infection is, in the SLIARS model, coupled with the contact parameter $\bc$, making up what is traditionally denoted $\bm{\beta}=\bdelta\circ\bc$ in classic models.

To compare results, we do need the reproduction number of the model considered here.
It takes the form \eqref{eq:R0-SLIARS}.
While the basic reproduction number $\R_0$ is equivalent for both models given an identical biological parameter set, their transient dynamics differ significantly due to their handling of initial exposure.
\begin{itemize}
    \item In the SEIARS model \eqref{sys:SEIARS}, contacts move individuals into an exposed state. 
    From here, only a fraction $\delta$ develop the infection and move to $I$ or $A$. 
    The remaining $1-\delta$ return to the susceptible pool after an average $1/\varepsilon$ time units.
    \item In the SLIARS model \eqref{sys:SLIARS}, only successful infections are modelled. 
    A contact immediately commits a fraction $\delta$ of the population to the latent state, where almost all individuals eventually become infectious (save for the ones not surviving natural death during the latent period), while the rest never leave the $S$ pool.
\end{itemize}

In the SEIARS framework, moving into the $\bE$ compartment temporarily removes individuals from the susceptible pool. 
This acts as a temporary shielding effect, preventing concurrent exposures. 
Biologically, this can be justified by rapid, localized innate immune activation (e.g., interferon responses in the respiratory mucosa) triggered by the initial sub-infectious contact, providing a brief window of non-specific viral resistance before the host returns to baseline susceptibility. 
Provided the exposure duration $1/\varepsilon$ remains short, this creates a plausible biological dampening effect absent in classical models.

To simplify comparisons, we consider here that parameters are equal across all $n$ Erlang chains $\bI$ and $\bA$ and corresponding $\bE$ compartments.
Furthermore, we operate in an epidemic context, by setting $b=d=\nu=0$, i.e., assuming there is no birth, death or loss of infection-acquired immunity.
The expression of $\R_0$ given by \eqref{eq:R_0} is the same in this context.
\subsubsection*{Shape of the outbreak (incidence)}
Considering the rate at which individuals become infectious (move into $I_1$ or $A_1$ compartments) in both models shows marked differences. 

\begin{figure}[htbp]
    \centering
    \begin{subfigure}{0.49\textwidth}
        \includegraphics[width=\linewidth]{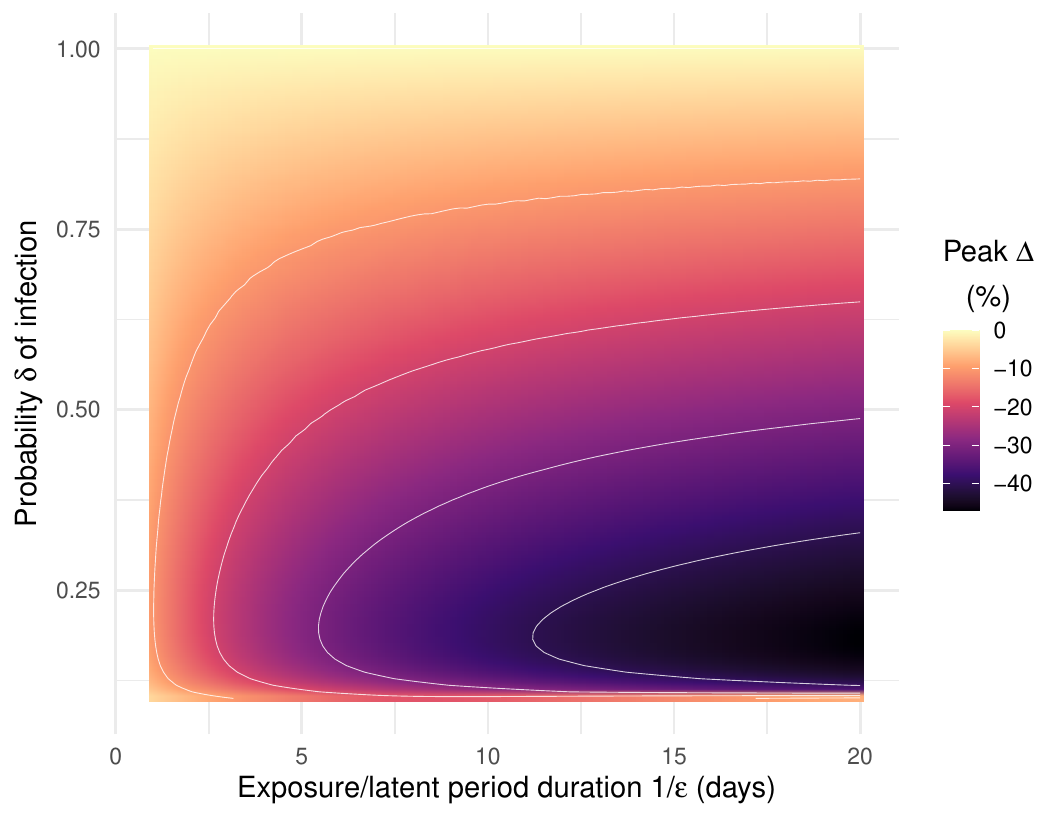}
        \caption{Difference in peak size}
        \label{fig:compare-SEIARS-SLAIRS-size}
    \end{subfigure}
    \begin{subfigure}{0.49\textwidth}
        \includegraphics[width=\linewidth]{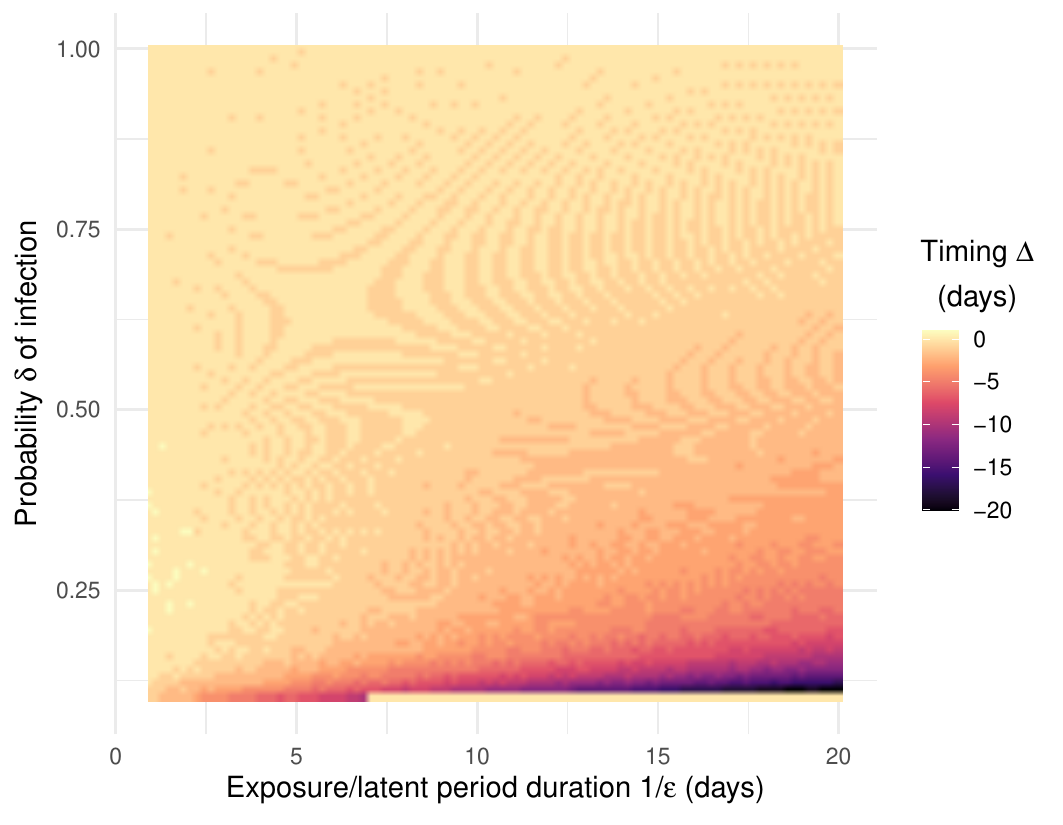}
        \caption{Difference in peak timing}
        \label{fig:compare-SEIARS-SLAIRS-timing}
    \end{subfigure} \\
    \begin{subfigure}[b]{0.48\textwidth}
        \includegraphics[width=\textwidth]{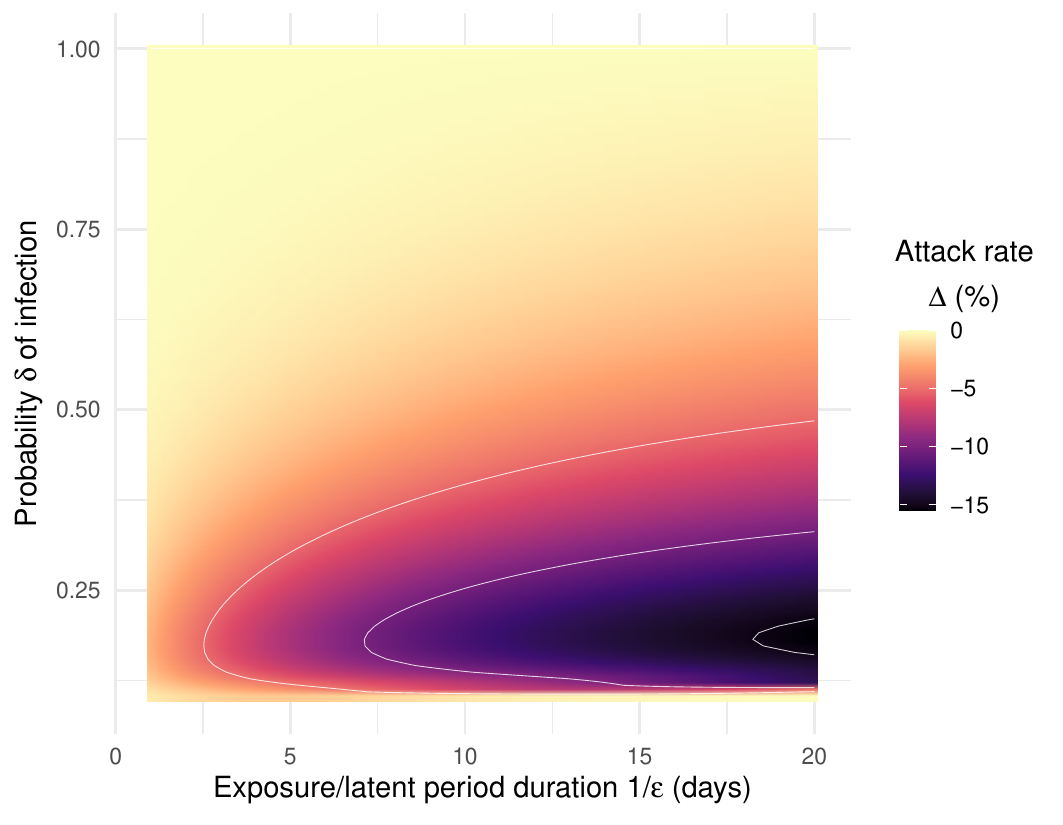}
        \caption{Attack rate}
        \label{fig:compare-SEIARS-SLAIRS-attack-rate}
    \end{subfigure}
    \caption{Comparison of the SEIARS and SLIARS models in terms of (a) peak incidence (\%), (b) peak timing (days) and (c) final size. All values given in percentages, with negative values indicating that the quantity is smaller in the SEIARS model than in the SLIARS model.}
    \label{fig:compare-SEIARS-SLAIRS}
\end{figure}

In Figure~\ref{fig:compare-SEIARS-SLAIRS}, we compare the magnitude of the peak size (Figure~\ref{fig:compare-SEIARS-SLAIRS-size}), peak timing (Figure~\ref{fig:compare-SEIARS-SLAIRS-timing}) and attack rate (Figure~\ref{fig:compare-SEIARS-SLAIRS-attack-rate}) between the SLIARS model taken as the baseline and the SEIARS model, when the mean duration of the exposure or latent period duration $1/\varepsilon$ and probability of infection $\delta$ are varied.
Note that here and elsewhere, the attack rate and various related ``total quantities'' are computed using the area under the curve of the relevant quantity over 1,500 days; the epidemic typically resolves much faster than that but we take a large time horizon so that even slow moving epidemics for extreme parameter values.
Negative differences are observed in both figures, meaning that the SEIARS model \eqref{sys:SEIARS} exhibits a lower peak incidence than the SLIARS model \eqref{sys:SLIARS}, that those peaks occur earlier and that the cumulative number of cases over the course of the epidemic is smaller.
Note that plotting the peak prevalence instead of the peak incidence shows a very similar picture (not shown).

This dynamic is driven by a ``buffer effect'': the existence of individuals in $\bE$ who eventually return to $S$ acts as a constant drain on the pool of potential new infections, leading to a smaller effective momentary outbreak.

\subsubsection*{Comparing apples and oranges}
We just established that the SLIARS model over-estimates the size of peaks compared to the SEIARS when using the same parameter values.
It is tempting to simply equate the parameter sets between the two models and continue the comparison this way. 
However, the differing internal structures of the two models means that the same nominal parameter values give rise to distinct macroscopic epidemic behaviours.
Thus, comparing them under identical parameter inputs can conflate structural differences with differences in the effective epidemic severity they simulate.

To address this, in the remainder of Section~\ref{sec:comparison-SEIRS}, we calibrate the models such that they align on a chosen observable epidemiological quantity.
While the basic reproduction number $\R_0$ is a common choice, in our case, as noted earlier, the expression for $\R_0$ in both the SEIARS and SLIARS models is identical.
Therefore, merely matching $\R_0$ does not sufficiently constrain the trajectories of the models.
Proper matching therefore requires using other measures that reflect the full nonlinear dynamics of the outbreak. 
Candidates include the peak infection prevalence, the final epidemic size, or the initial exponential growth rate. 
Here, we chose to match the \emph{initial exponential growth rate} ($r$) of the epidemic. 

Specifically, for a given set of parameters ($1/\varepsilon$ and $\delta$ in most examples below), we treat the SEIARS model as the ``true'' baseline and simulate its early exponential growth trajectory. 
We then numerically optimize the baseline contact rates ($c$) in the corresponding SLIARS model to minimize the difference in this initial growth rate. 
By doing so, we mimic the real-world public health practice of fitting baseline models to early outbreak data. 
Any remaining macroscopic discrepancies---such as differences in the peak prevalence or total final size---are therefore strictly attributable to the structural differences in how the prevalence-dependent exposure buffer naturally slows the SEIARS epidemic, allowing us to highlight the projection bias inherent in classical models.

\subsubsection*{The holding effect and isolation burden}
\begin{figure}[htbp]
    \centering
    \begin{subfigure}[b]{0.48\textwidth}
        \includegraphics[width=\textwidth]{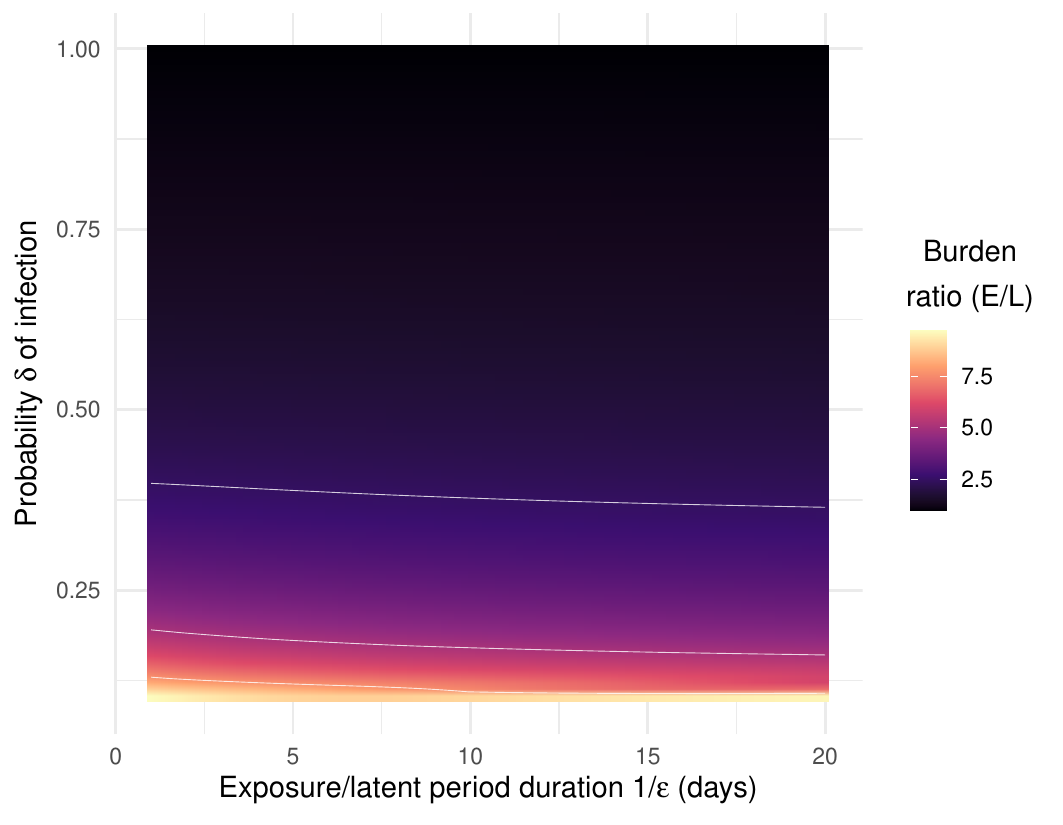}
        \caption{Burden ratio}
        \label{fig:compare-SEIARS-SLAIRS-isolation-burden-ratio}
    \end{subfigure}
    \caption{Comparison of the characteristics influencing isolation.}
    \label{fig:compare-SEIARS-SLAIRS-isolation}
\end{figure}

In the Introduction, one of the main criticisms we lobbed against classic ``exposure = latency'' models was related to a potential under-estimation of the cost of non-pharmaceutical interventions such as isolation or quarantine. While classical models can explicitly incorporate tracing by manually bolting on quarantined susceptible ($S_Q$) compartments, their base structures inherently erase the underlying biological population of ``near-misses.''

Using the same variation of parameters as in Figure~\ref{fig:compare-SEIARS-SLAIRS}, we consider in Figure~\ref{fig:compare-SEIARS-SLAIRS-isolation} two different quantities that illustrate this problem.
In Figure~\ref{fig:compare-SEIARS-SLAIRS-isolation-burden-ratio}, we compute the ratio of the area under the curves $\bE$ and $\bL$.
We see that when the probability $\delta$ of infection is low, the ratio of areas under the curves $\mathsf{AUC}(\bE)/\mathsf{AUC}(\bL)$ increases greatly: over the course of an outbreak, there are up to ten times more individuals ``holding'' in the $\bE$ compartments as there are in the $\bL$ compartments.
This acts as an idealized proxy for the true cost of isolation, since tracing and isolation policies apply to all exposed individuals, not just those who progress on to infected compartments. While classical epidemiological models can be extended with explicit quarantine compartments ($S_Q$) to simulate tracing, their base structure inherently erases the underlying biological population of ``near-misses''. By natively tracking the $\bE \to S$ loop, the SEIARS framework provides a built-in upper bound for the volume of false-positive individuals who would naturally be captured by a contact-tracing dragnet.

One could argue that the SEIARS model predicts a smaller incidence and final size, so the lower societal cost of infection could counterbalance the higher cost of isolation.
However, this is forgetting that while the values of $\R_0$ agree for the SEIARS and SLIARS models for equal parameters, the SEIARS model would probably show a larger value of $\R_0$ if parametrised using incidence data.
See for instance the next section on \emph{Mortality and projection bias}.

This difference in the number of individuals in the $\bE$ and $\bL$ compartments has a critical public health implication: standard latent models (SLIARS) drastically underestimate the social and economic cost of isolation and contact tracing because they only track individuals who will eventually get sick, ignoring the large volume of exposed individuals who will be quarantined but ultimately return to $S$. 
Correspondingly, the SEIARS $S$ pool drops rapidly as it is absorbed into $E$, then rebounds as people return to $S$ without infection, creating a unique population ``U-turn'' dynamic absent in classical models.

\subsubsection*{SLIARS models over-estimate mortality}
We now investigate the frequent overestimation of total mortality in early stages of a nascent epidemic or pandemic through the lens of the exposure model \eqref{sys:SEIARS}.
Matching the SEIARS and SLIARS models as explained earlier, we extract the cumulative disease-induced deaths from both simulations and plot the relative overestimation error produced by the SLIARS model. 
The results, visualised in Figure~\ref{fig:mortality-experiment}, show that even though the SLIARS model is fitted to the early growth data, it systematically overestimates total mortality. 
Under standard baseline assumptions, this overestimation error sits roughly at 17\%, but aggressively escalates when tracking a highly contagious disease profile where the true progression probability $\delta_{\text{true}}$ is relatively low.

\begin{figure}[htbp]
    \centering
    \begin{subfigure}[b]{0.48\textwidth}
        \includegraphics[width=\textwidth]{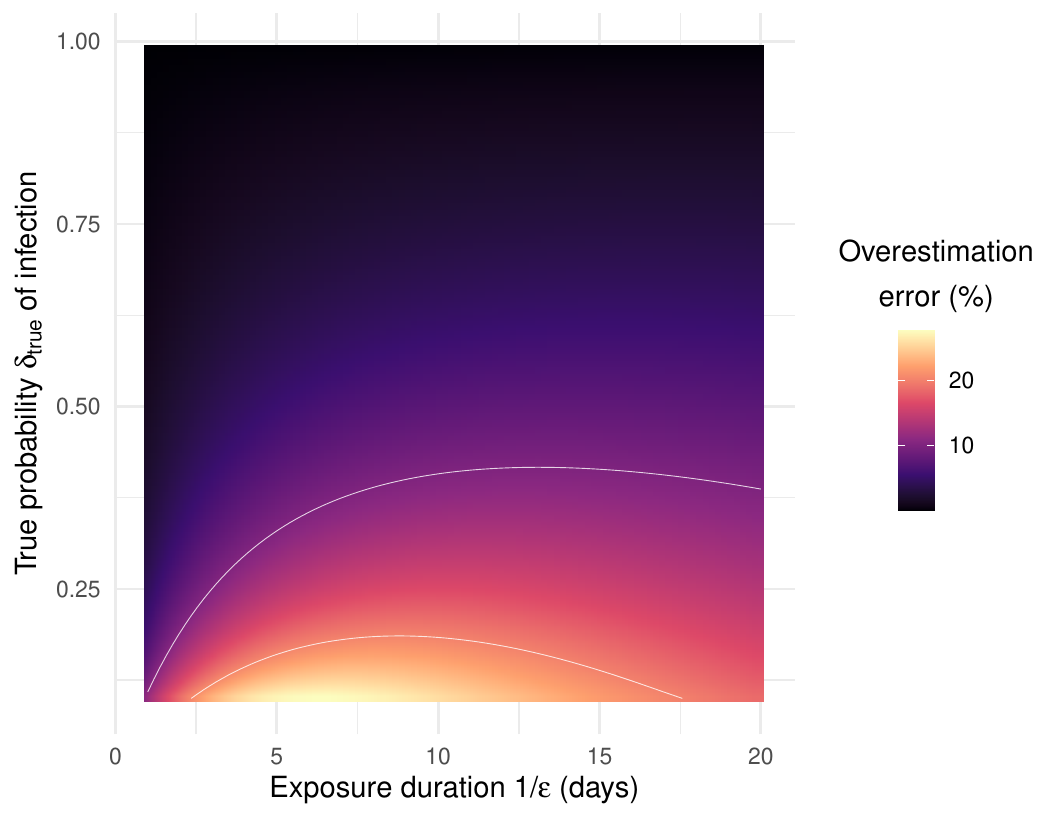}
    \end{subfigure}
    \hfill
    \begin{subfigure}[b]{0.48\textwidth}
        \includegraphics[width=\textwidth]{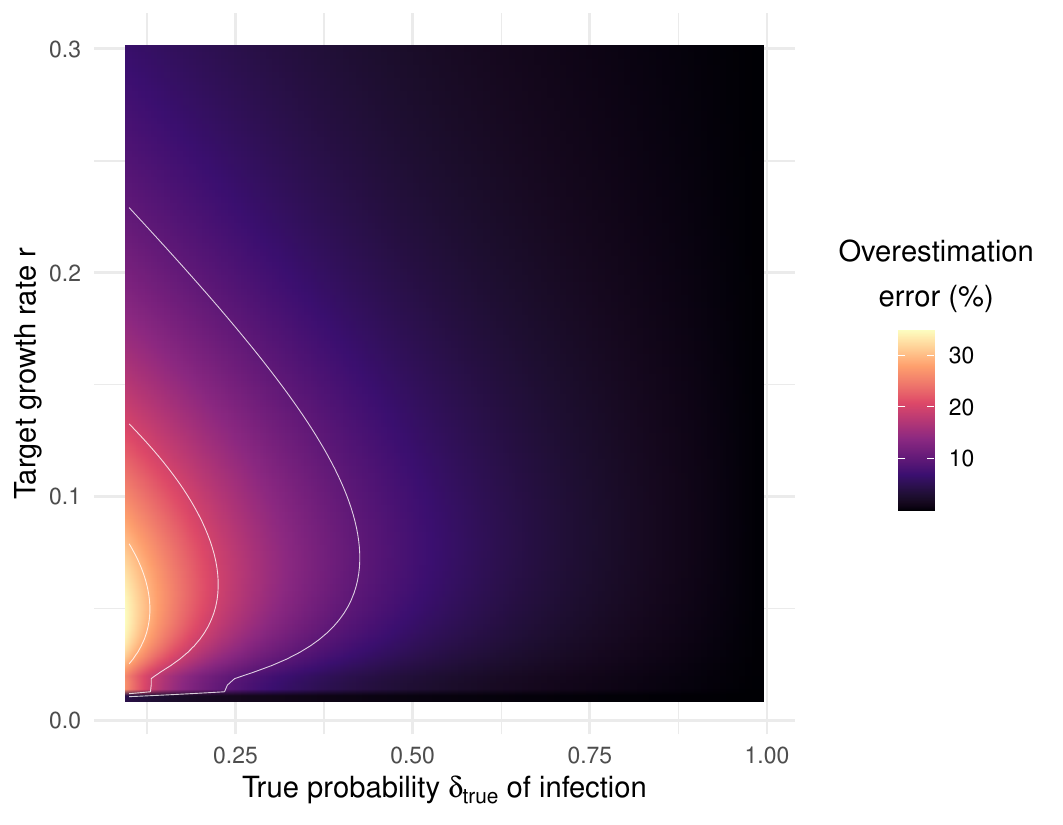}
    \end{subfigure}
    \caption{Mortality bias arising when a SLIARS model is matched to an SEIARS model using the initial exponential growth rate $r$.}
    \label{fig:mortality-experiment}
\end{figure}

This difference is a result of the exposed state acting as a temporary sink for susceptibles. 
This reduces the effective density of available susceptibles faster than actual confirmed infections would suggest, causing the outbreak to hit the herd immunity threshold sooner and naturally burn out with fewer total deaths.

\subsubsection*{Contact tracing efficiency}
The ``social cost'' of the classical exposed state is further quantified by measuring how many people are isolated unnecessarily. 
For a low $\delta$ threshold in the SEIARS model, a massive fraction of isolated individuals represent ``false positives'' who never actually develop the disease.
This results in millions of person-days spent in isolation by false positives, a socioeconomic effect that is structurally obscured in classic baseline SLIARS models. While those models can simulate tracing via explicit quarantine extensions, they typically apply tracing efforts to a network derived solely from truly infected individuals, missing the vast biological reservoir of natural false positives.

To illustrate the cost in the SEIARS model, Figure~\ref{fig:contact-tracing-efficiency} shows a sweep across the probability $\delta$ of infection. 
For each simulation, we calculat the total integrated time spent in the isolation compartment ($E$) by individuals who ultimately return to the susceptible state. 
As shown in Figure \ref{fig:contact-tracing-efficiency}, operating contact tracing at lower probabilities of infection triggers a severe, non-linear explosion in the societal isolation burden. 
Notably, for $\delta$ values below the epidemic threshold ($\R_0 < 1$, marked by the dotted red line), the outbreak fails to sustain itself, and the societal isolation burden remains negligible. However, immediately upon crossing the threshold, the exponential growth of the epidemic forces the false-positive isolation burden to spike dramatically, until the increasing true-positive rate and finite population size eventually constrain it.
\begin{figure}[htbp]
    \centering
    \includegraphics[width=0.6\textwidth]{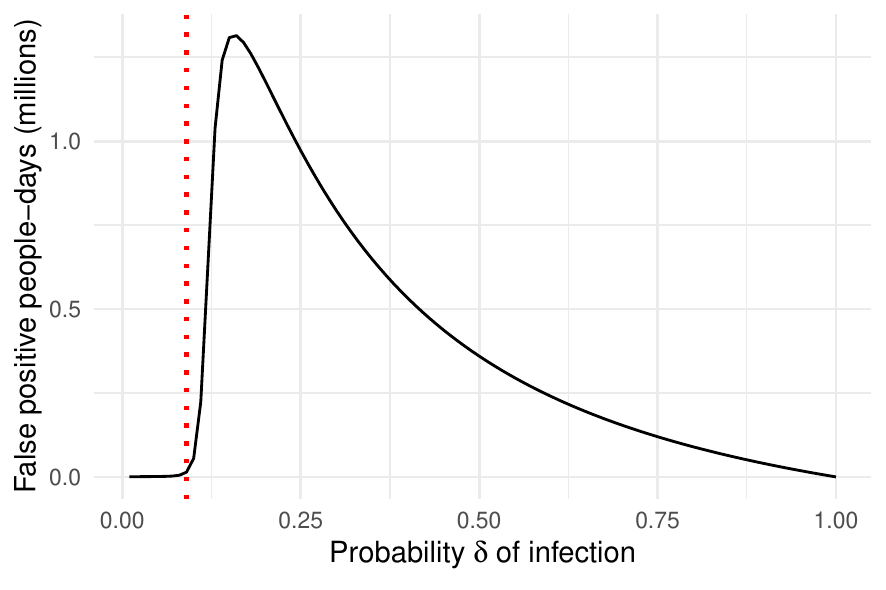}
    \caption{Societal cost of false-positive isolation as a function of infection probabilities $\delta$. 
    The curve represents the total millions of person-days spent in isolation by individuals who never actually become infectious. 
    The red line marks the epidemic threshold $\R_0=1$.}
    \label{fig:contact-tracing-efficiency}
\end{figure}

To quantify this mathematical drag structurally, we execute a high-resolution simulation sweep computing the operational efficiency of contact screening under the standard SLIAR assumption versus our full model. 
By indiscriminately applying contact tracing effort to false-positive contacts (defined by a true-positive tracking probability $\delta$), we match the exact same true-positive exposure flow into the compartment.
As shown in Figure \ref{fig:contact-tracing-comparison}, modeling contact tracing using standard isolation compartments structurally overestimates the efficiency per unit effort (isolation-days enforced). 
This discrepancy becomes aggressively more pronounced at lower testing specificity/true-positive rates ($\delta$), while simultaneously exhibiting a shielding non-linearity: isolating massive networks of false positives paradoxically buffers the true severity under high baseline transmission settings ($c_{scale}$), while causing efficiency outcomes to plummet further relative to classic expectations.

\begin{figure}[htbp]
    \centering
    \begin{subfigure}[b]{0.48\textwidth}
        \centering
        \includegraphics[width=\textwidth]{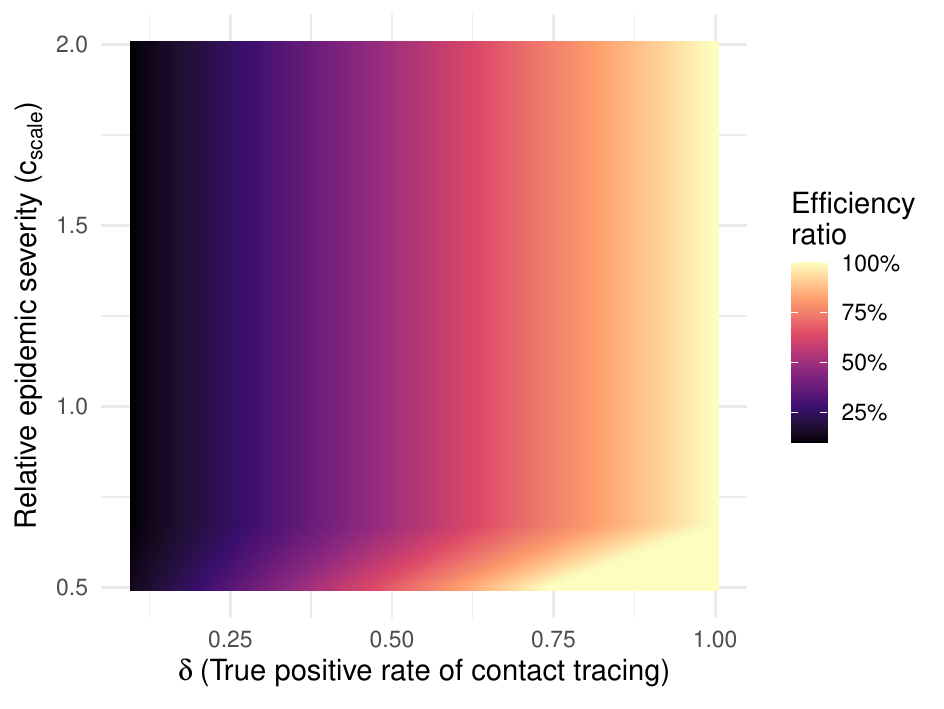}
        \caption{Efficiency}
        \label{fig:contact-tracing-comparison-ratio}
    \end{subfigure}
    \hfill
    \begin{subfigure}[b]{0.48\textwidth}
        \centering
        \includegraphics[width=\textwidth]{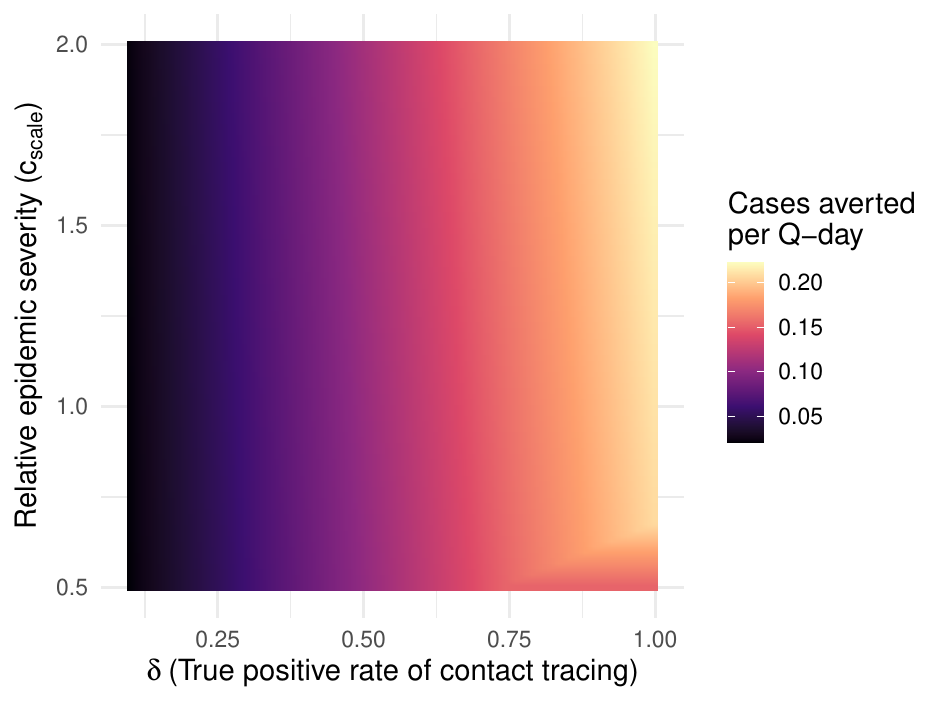}
        \caption{Averted cases}
        \label{fig:contact-tracing-comparison-absolute}
    \end{subfigure}
    \caption{Impact of the true positive tracing rate ($\delta$) against relative epidemic severity ($c_{scale}$). Classical computations are rigidly chained to outcomes representing the extreme right boundary ($\delta=1.0$), thereby severely discounting the ballooning network effort of isolating non-infectious individuals. 
    (a) Ratio of True model efficiency versus the expected efficiency modelled under standard SLIAR assumptions.
    (b) Absolute cases averted structurally per quarantine-day across the full mathematical model.}
    \label{fig:contact-tracing-comparison}
\end{figure}



\subsection{Role of the number of compartments in Erlang chains}
\label{sec:Erlang-on-or-off}

\begin{table}[htpb]
    \centering
    \begin{tabular}{lcc}
        \toprule
        Criteria & Change $n=1\to n=20$ & Change $n=2\to n=20$ \\
        \midrule
        $\R_0$   & 2.5\% & 1.2\% \\
        Peak prevalence   & 227.2\% & 102.9\% \\
        Peak time         & -62.2\% & -43.4\% \\
        Final deaths      & 8.4\% & 4.8\% \\
        \bottomrule
    \end{tabular}
    \caption{Percentage change in key epidemiological quantities as the number of stages $n$ changes.
    Here, $c^I = 0.5$ and $c^A = 0.25$.}
    \label{tab:erlang_n_variation}
\end{table}

We use a generalised Erlang distribution to model the total time spent in infectious compartments.
We briefly investigate here the effect of changing the number of compartments in these Erlang chains.
In Table~\ref{tab:erlang_n_variation}, we see the effect of changing $n$ from 1 to 20 on various quantities, highlighting the influence of that parameter.
We also show the changes when the strictly exponential case is excluded, i.e., $n$ changes from 2 to 20.
Reducing the respective variance of the sojourn time (by increasing $n$) synchronizes individual transitions across the population, leading to a more concentrated and intense outbreak peak compared to an exponential ($n=1$) assumption.

Notice that the reproduction number $\R_0$ increases as a function of $n$. 
It is easy to find parameter ranges for which $\R<1$ for some $n< n_c$ and $\R_0\geq 1$ for $n\geq n_c$.
For instance, lowering $c^I$ to 0.25 with other parameters as used to create Table~\ref{tab:erlang_n_variation}, we find $n_c=4$.
Now, $n$ is typically fixed by knowledge about time of sojourn in the infected compartments, so $n$ is not \emph{per se} a bifurcation parameter, but this is a fact useful to keep in mind.


\subsection{Generalised Erlang durations and viral load mapping}
\label{sec:viral-load-mapping}
When formulating the model in Section~\ref{sec:math-model}, we remarked that the $\bI$ and $\bA$ chains do not have strict Erlang distributions of sojourn times.
We explore this in more detail here.

First, regarding sojourn times not being Erlang, note the following.
Let $\mu_{\min}=\min_i\mu_i$ and $\mu_{\max}=\max_i\mu_i$.
Then the time spent in the chain of $\bI$ compartments is bounded below and above by Erlang distributions with shape parameters $n$ and scale parameters $\gamma_I+\mu_{\min}+d$ and $\gamma_I+\mu_{\max}+d$, respectively.
This is due to closure under convolutions of the usual stochastic order \cite[Theorem 1.A.3]{shaked2007stochastic}.
The resulting distribution of sojourn times in $\bI$ is called \emph{hypoexponential} \cite{trivedi2016probability} or, in applied contexts, \emph{generalised Erlang} \cite{kluppelberg2020explicit}.
For details in the mathematical epidemiology context, see \cite{hurtado2021building}.
Similarly, the time spent in the $\bA$ chain is bounded above and below by Erlang distributions with shape parameters $n$ and scale parameters $\gamma_A+\tau_{\min}+d$ and $\gamma_A+\tau_{\max}+d$, respectively, where $\tau_{\min}=\min_i\tau_i$ and $\tau_{\max}=\max_i\tau_i$.

\begin{figure}[htbp]
    \centering
    \begin{subfigure}[b]{0.48\textwidth}
        \centering
        \includegraphics[width=\textwidth]{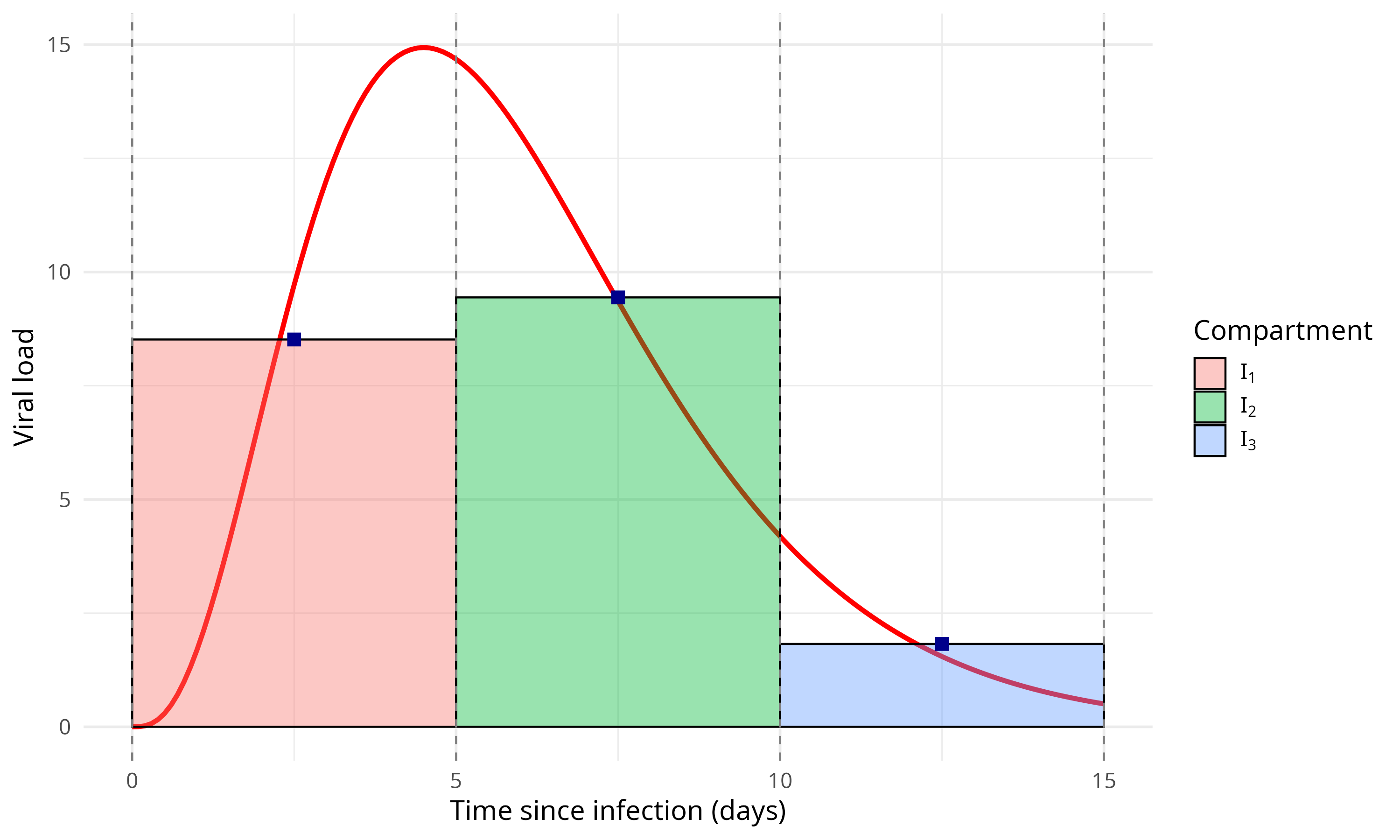}
        \caption{Constant $\gamma$ (Erlang)}
        \label{fig:viral-erlang}
    \end{subfigure}
    \hfill
    \begin{subfigure}[b]{0.48\textwidth}
        \centering
        \includegraphics[width=\textwidth]{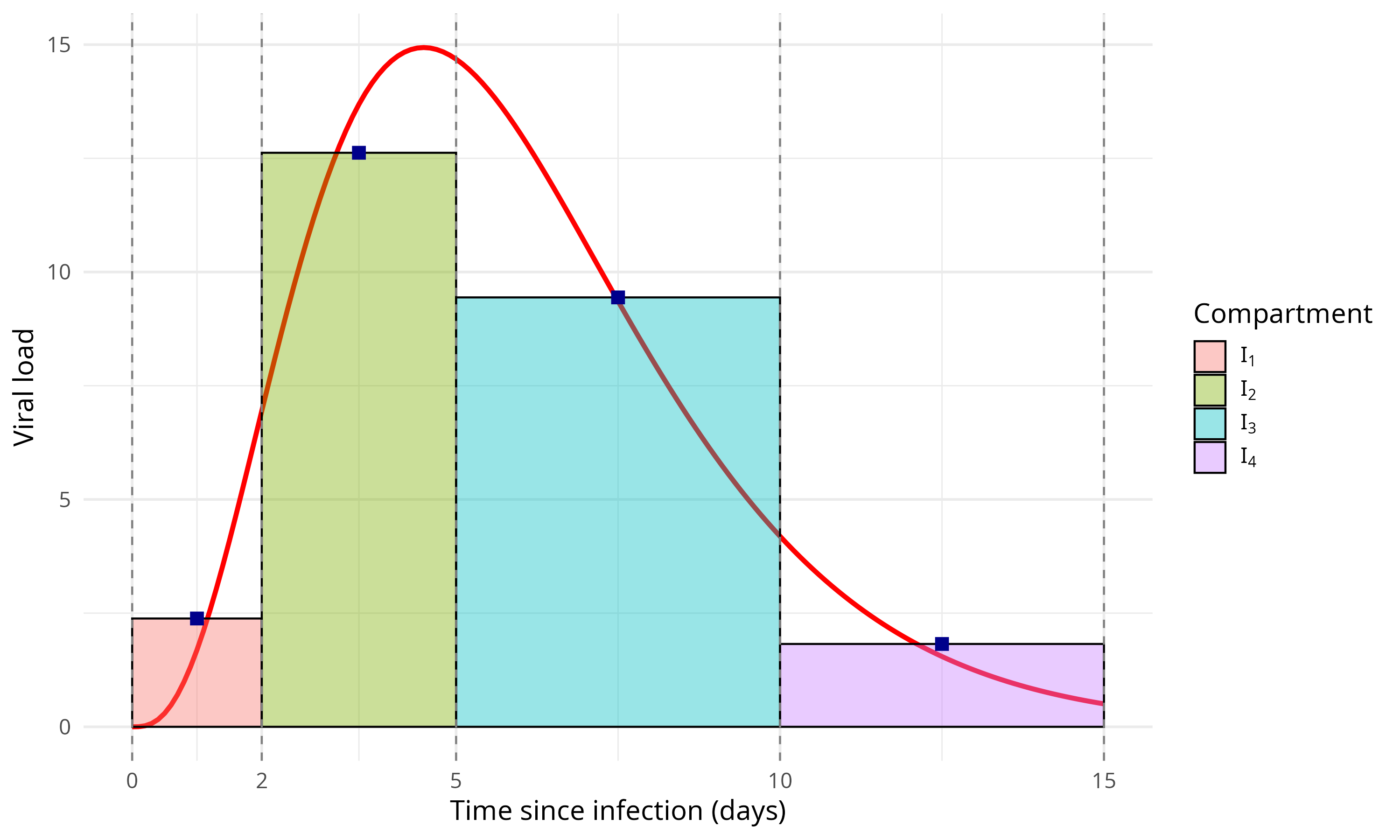}
        \caption{Varying $\gamma_k$ (generalised Erlang)}
        \label{fig:viral-varying}
    \end{subfigure}
    \caption{Mapping a continuous viral load function (red curve) to discrete compartmental parameters. 
    In (a), the Erlang case enforces equal residency times for all stages. 
    In (b), a generalised Erlang has varying compartment lengths.}
    \label{fig:viral-load-subfigures}
\end{figure}

By the same argument, we could also individualise the $\gamma_I$ and $\gamma_A$ and retain a generalised Erlang structure.
However, due to the difficulty in obtaining reliable data on sojourn times in infected compartments, we keep all $\gamma_I$ and $\gamma_A$ equal throughout the $\bI$ and $\bA$ chains, respectively.
This makes fitting to potentially observed infectious stage durations easier; see, e.g., the Ebola Virus Disease example in \cite[Section 3.1.2]{Arino2020a}.
Note that the added complexity from individualising $\gamma_I$ and $\gamma_A$ is minor; the mathematical analysis carries through without any issue and the problem is mostly notational.

Let us first consider the idealised case of equal $\gamma_X$.
Figure~\ref{fig:viral-load-subfigures} illustrates the translation of a continuous viral load curve into the corresponding discrete infection parameter $\delta_k^I$ assigned to each sub-compartment $I_k$. 
Because the residency time in each compartment is exponentially distributed, the expected entry time of an individual into compartment $I_k$ is the sum of the mean durations of all preceding $k-1$ compartments. 
As depicted in Figure~\ref{fig:viral-erlang}, \emph{ceteris paribus}, enforcing an identical transition parameter $\gamma_I$ across all compartments guarantees an equal average residency time uniformly distributed along the period of infection (the ``pure'' Erlang case).

On the other hand, relaxing the uniformity assumption to support varying phase durations $\gamma_{Xk}$, as illustrated in Figure~\ref{fig:viral-varying}, gives a generalised Erlang residency distribution, which can more flexibly track a target viral load. (In Figure~\ref{fig:viral-erlang}, one would have to substantially increase the number of compartments to better track the viral load.)
Another situation where varying $\gamma_{Xk}$ might be justified is when some durations are known.
For instance, in \cite{ArinoBoelleMillikenPortet2021}, the latent chain $\bm{L}$ is divided into a compartment with average sojourn time of 3 days and one with average sojourn time of 1 day, since in the case of COVID-19, latency was on average of 4 days with the last day being potentially infectious, which the model in \cite{ArinoBoelleMillikenPortet2021} considers.

To finish on sojourn times, note that one has to keep in mind that even in the Erlang case, residence times in the chain are only average, so two individuals surviving through the chains may experience very different total residence time, making the notion of chronological time somewhat arbitrary.
Synchronisation between $\bI$ and $\bA$ compartments is thus by position in the chains, not by actual chronological time.
If more precision is critically needed or if much finer grained data is available, then partial differential equations models including age-of-infection information should be used.

\subsection{Effect of infection parameter skew}
\label{sec:parameter-skew}
We now turn to the role of the infectivity $\delta$ in relation to these generalised Erlang chains, i.e., we dig more into the role of the red curve in Figure~\ref{fig:viral-load-subfigures}, investigating how the specific distribution of infection probability $\delta$ across the stages of the exposed compartments affects the epidemic, given a constant mean $\delta$.
Figure~\ref{fig:skew-illustration} shows this skewing mechanism in the case where there are $n=3$ compartments in the chains: a constant mean infectivity $\delta$ is differently distributed across the disease stages depending on the chosen peak location $\sigma\in[1,3]$, where $\sigma=1$ means the peak occurs at the start of the stage sequence (i.e., is exponentially distributed) while $\sigma=3$ means it occurs at the very end.
To simplify the situation, we only consider $n=3$ here.

\begin{figure}[htbp]
    \centering
    \includegraphics[width=\textwidth]{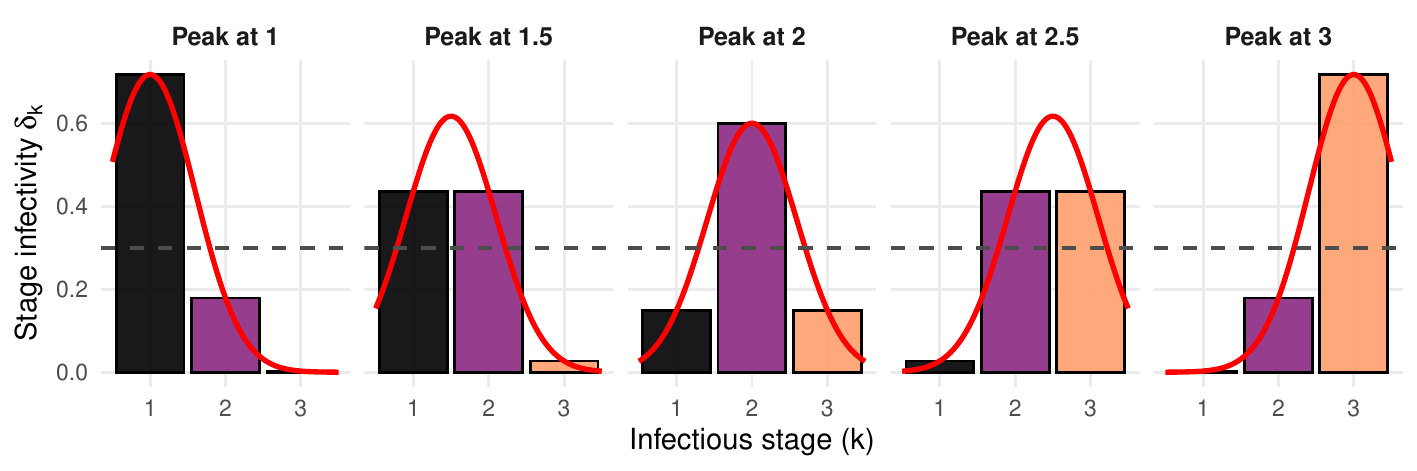}
    \caption{Illustration of the skew function, showing how infectivity risk is distributed across 3 stages for various peak locations ($\sigma$) while maintaining a constant mean $\delta = 0.3$.}
    \label{fig:skew-illustration}
\end{figure}

To explore different temporal profiles of infectivity while keeping the overall total transmission potential approximately constant, we distribute a baseline total infectivity across $n$ infectious stages using a Gaussian weighting function. Let $\delta_{\text{mean}}$ be the baseline mean infectivity per stage, corresponding to a total baseline infectivity of $n \delta_{\text{mean}}$. For a given peak location parameter $\sigma$ (which controls which stage is the most infectious) and a fixed spread $s$, the unnormalized weight for stage $k \in \{1, \dots, n\}$ is defined as
\[
w_k = \exp\left(-\frac{(k - \sigma)^2}{2s^2}\right).
\]
We normalize these weights so they sum to 1 and then scale them by the total baseline infectivity. 
Thus, the target infectivity for stage $k$ is given by
\[
\delta_k = \frac{w_k}{\sum_{j=1}^n w_j} \; n \delta_{\text{mean}}.
\]
Finally, to respect the constraint that a stage's infection probability cannot exceed 100\%, we cap the resulting $\delta_k$ at $1.0$.

The plots in Figure~\ref{fig:heatmap-skew} show the effect of varying $\delta$ and $\sigma$.
As indicated by the incidence heatmap in Figure~\ref{fig:heatmap-skew}, a late progression benefit emerges: for a fixed average $\delta$, shifting the infection risk toward the later stages of the exposure period reduces the peak prevalence and the final population attack rate.

\begin{figure}[htbp]
    \centering
    \begin{subfigure}[b]{0.32\textwidth}
        \includegraphics[width=\textwidth]{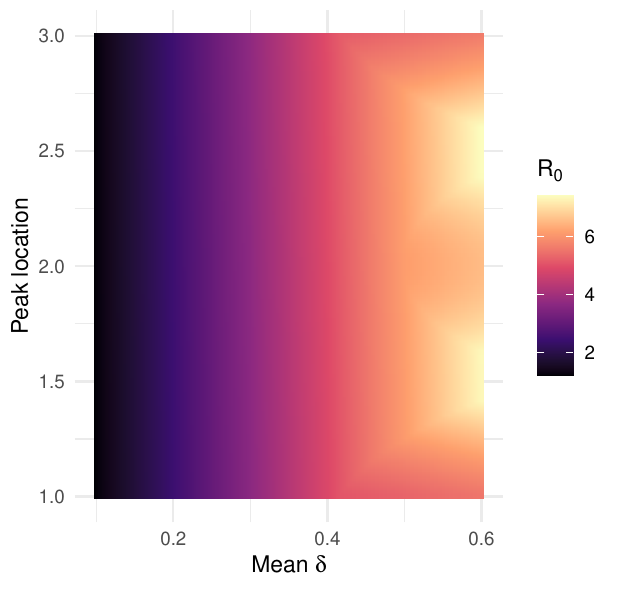}
        \caption{$\R_0$}
        \label{fig:heatmap-skew-R_0}
    \end{subfigure}
    \hfill
    \begin{subfigure}[b]{0.32\textwidth}
        \includegraphics[width=\textwidth]{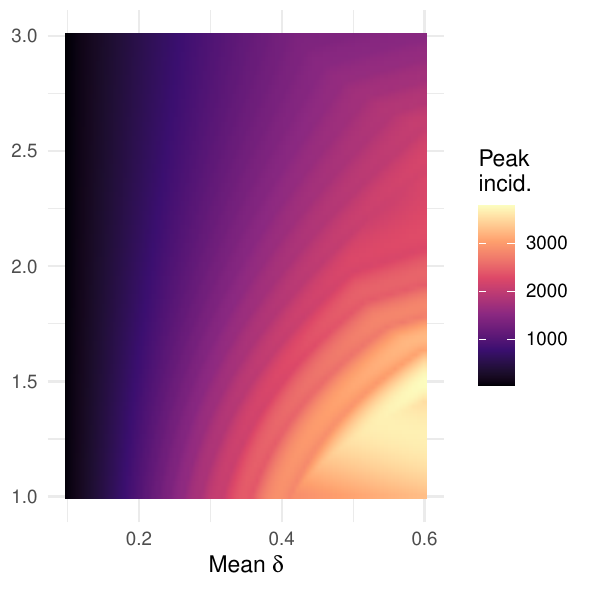}
        \caption{Peak incidence}
        \label{fig:heatmap-skew-incidence}
    \end{subfigure}
    \hfill
    \begin{subfigure}[b]{0.32\textwidth}
        \includegraphics[width=\textwidth]{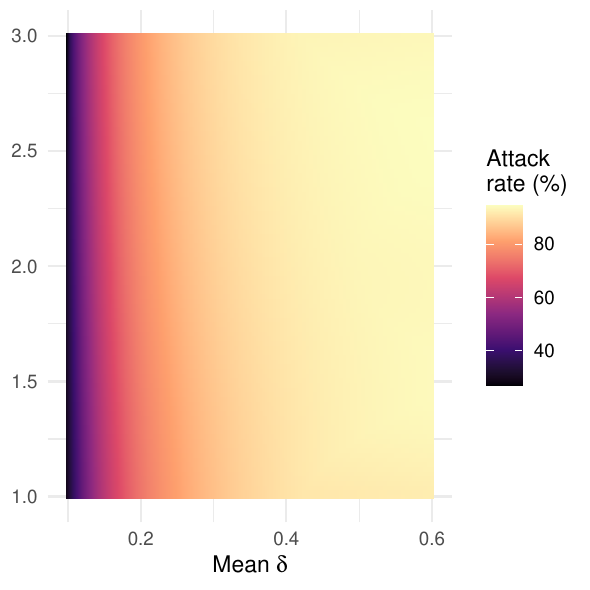}
        \caption{Attack rate}
        \label{fig:heatmap-skew-attack-rate}
    \end{subfigure}
    \caption{Effect of varying the distribution of a fixed $\delta$ across 3 stages.}
    \label{fig:heatmap-skew}
\end{figure}

This indicates that the shape and orientation of the internal parameter distribution can operate as a deciding factor in whether a localized, marginal disease outbreak formally crosses the epidemic threshold.

\section{Discussion}
\label{sec:discussion}
In this work, we lay the basis for taking into consideration exposure to a pathogen not automatically followed by development of the infection stemming from that pathogen, not with a simple fraction of whatever contact function as is used in classic epidemiological models.

We could have considered a simpler model with $n=1$ compartments in the Erlang chains, i.e., a single of each the $E$, $I$ and $A$ compartments.
However, ``qui peut le plus peut le moins'' (who can do more, can do less) and the analysis and ideas here naturally carry forward to the $n=1$ case.
Also, retaining the full dimensionality of the model allows to use differential progression to infection based on some proxy for age-of-infection of the potentially infecting contact.
Modelling infectivity of an individual dependent on how long they have carried the pathogen for is not new: the seminal paper of Kermack and McKendrick \cite{kermackmckendrick1927} is best remembered for the three dimensional SIR model that bears their names. 
That model is, however, a simplification of the age-of-infection structured model they introduce.
Later, \cite{lajmanovich1976deterministic} formalised group models, i.e., models where the infected population is heterogenous and in which each group has a different force of infection acting on the susceptible population.
Moving in the direction of the temporal structuration of groups, \cite{jacquez1988modeling} considered a complex stage-dependent model for HIV. 
A simplified version of that model became the well-known \cite{Hymanetal1999}.
Similarly, co-infection models bear some similarity with our model, in that individuals enter the infected compartment related to the strain they acquire.
All these models, though, retain the idea that a fraction of the contacts with whatever type of infectious individual is involved moves into the corresponding infected compartment.

Much of the computational analysis we carry out concerns the comparison with an equivalent ``classic'' model where a fraction of the contacts progress into a latent stage.
Utilising classic SLIARS mechanics when the true human dynamics follow SEIARS patterns consistently leads to
\begin{enumerate}
    \item Overestimating the peak case burden on the healthcare system.
    \item Vastly underestimating the socio-economic burden of quarantine and contact tracing.
    \item Systematically overestimating the final size and total mortality of the overall outbreak.
\end{enumerate}
This has important implications in terms of public health policy formulation. 

The model we introduce has limitations, of course.
Firstly, we assume that the first contact that a susceptible individual makes determines the characteristics of their exposed period, both in terms of likelihood of clearing the infection and outcome of the infection  if they do not clear it (symptomatic or asymptomatic).
In ``real life'', susceptible individuals make multiple contacts with infectious individuals and it is likely that repeated contacts during the exposed period increase the probability of developing an infection.
Secondly, the number of parameters involved is potentially huge and many of them are difficult, if not impossible, to estimate.
The model is clearly not identifiable, neither structurally nor in practice.
Thirdly, lacking here is vaccination, which, when available, plays a very important role in determining progression to infection.
The latter is the object of ongoing work for Part 2, which investigates what happens if vaccination is added to the model.

\subsection*{Acknowledgments}
JA is supported in part by an NSERC Discovery Grant.

\appendix

\counterwithin{equation}{section}
\counterwithin{figure}{section}
\counterwithin{table}{section}
\setcounter{equation}{0}
\setcounter{theorem}{0}
\renewcommand{\theequation}{\thesection.\arabic{equation}}
\renewcommand{\thetheorem}{\thesection.\arabic{theorem}}
\setcounter{section}{0}

\section*{Appendices}
\section{Proofs in the mathematical analysis}
\label{app:proofs}

\subsection{The model is well-posed (Theorem~\ref{th:positivity})}
\label{app:proofs-well-posedness}
The proof of Theorem~\ref{th:positivity} proceeds in three steps.

\emph{Step 1: Positivity --}
Let $\bX$ be the solution of \eqref{sys:SEIARS} corresponding to the initial conditions $\bX(0)\gg \bzero$.
Then, there exists some $t_0>0$ such that $\bX(t)$ remain non-negative for all $t\in \, (0,t_0)$.

Assume that $t_0<+\infty$.
Then there exists $t_1>t_0$ such that at least one component of $\bX$ satisfies $\{X_j(t_1)=0,\dot{X}_j(t_1)<0\}$, for an existing $j\in \{1,\dots,4n+2\}$.
It follows that the set
\begin{align*}
\mathcal{J}=&\{ t_1>t_0|\,\exists j\in\{1,\dots,2n+3\},\,\{X_j(t_1)=0,\dot{X}_j(t_1) <0\}\}
\end{align*}
is non-empty.
Let
\[
t^\star=\min\limits_{t_1}\left\{\mathcal{J}\right\}.
\]
Without loss of generality, assume that $S(t^\star)=0$ and $\dot{S}(t^\star)<0$, i.e., $j=1$.
With $S(t^\star)=0$, \eqref{sys:SEIARS-dS} gives
\[
\dot S(t^\star)=b+(1-\delta(\bY_{\text{exp}}(t^\star)))\varepsilon E(t^\star)
+\nu R(t^\star) >0,
\]
which implies that $S(t^\star)>0$.
This is absurd since $\dot{S}(t^\star)<0$.

The same approach is used for the other state variables to show that it is impossible to have
$\{X_j(t_1)=0,\dot{X}_j(t_1)<0\}$, for an existing $j\in \{1,\dots,4n+2\}$.
Then $t_1=+\infty$. Therefore, the trajectories of the solutions of system \eqref{sys:SEIARS} remain positive for all $t\in [0,+\infty[$.

\emph{Step 2: Boundedness --}
Adding all equations, the dynamics of total human population satisfies
\begin{equation}\label{d_to}
\dot{N}=b-d N-\left\langle\mu,I\right\rangle\leq b-d N(t).
\end{equation}
Applying the Gronwall inequality to \eqref{d_to} with the initial condition $N(0)$ gives
\[
N(t)\leq \dfrac{b}{d}+\left(N(0)-\dfrac{b}{d}\right)e^{-d t}=\dfrac{b}{\mu}\left(1-e^{-d t}\right)+N(0)e^{-d t},\quad \forall t \geq 0.
\]
This implies that $N(t)\leq N(0)+{b}/{d}$ for all time $t\geq 0$.
Since state variables are non-negative (Step 1), the boundedness of the total population \eqref{eq:total} implies that all state variables are also bounded.
Therefore, combining Steps 1 and 2, the result about the positivity and boundedness of solutions in \eqref{om} is proved.

\emph{Step 3 --} Here, one shows how the system \eqref{sys:SEIARS} admits a unique maximal solution. For this, we prove that the right part is less than a bilinear function. We will show the force of infection $\lambda$ can not be infinity.

Firstly, there exists a constant $\eta_0>0$ such that the dynamic of the total population verifies $N(t)\geq \eta_0, \forall t>0$.
Suppose that $\forall \eta_0>0$, there exists $t_0>0$ such that $N(t_0)< \eta_0$.
Using the dynamic of $N$ in \eqref{d_to} and the fact that $I_i(t),A_i(t)\leq N(t)\; (i=1,\dots,n), \,\forall t\geq 0$, it follows that
\begin{equation}\label{into}
\dot{N}(t)\geq b-\left(d+\sum\limits_{1=1}^{n}\mu_i\right)N(t).
\end{equation}
Applying the Gronwall inequality to \eqref{into} with the initial condition $N(0)$ gives that
\[
N(t)\geq \dfrac{b}{d+\sum\limits_{1=1}^{n}\mu_i}+\left(N(0)-\dfrac{b}{d+\sum\limits_{1=1}^{n}\mu_i}\right)e^{-\left(d+\sum\limits_{1=1}^{n}\mu_i\right)t},\quad \forall t \geq 0.
\]
At $t=t_0$, one has $N(t_0)\geq \eta_0$, where
\[
    \eta_0=\dfrac{b}{d+\sum\limits_{1=1}^{n}\mu_i}\left(1-e^{-\left(d+\sum\limits_{1=1}^{n}\mu_i\right)t_0}\right)+N(0)e^{-\left(d+\sum\limits_{1=1}^{n}\mu_i\right)t_0},
\]
a contradiction since we assumed that $N(t_0)< \eta_0$ for all $\eta_0> 0$.
It follows that $\exists \eta_0>0,\,N(t)\geq \eta_0, \forall t>0$.

This implies that the force of exposure \eqref{eq:force-of-exposure} is bounded as
\[
\lambda\leq \dfrac{\left\langle\bY,\left(c_1^I,\ldots,c_n^I,c_1^A,\ldots,c_n^A\right)\right\rangle}{\eta_0}.
\]
Then, the right part in system \eqref{sys:SEIARS} is less than a bilinear function (which is Lipschitz naturally) and consequently in $C^\infty(\Omega)$. This completes the proof of Theorem \ref{th:positivity}.

\subsection{$\R_0$ and stability of the disease-free equilibrium}
\label{app:proofs-R_0-LAS}
To determine the basic reproduction number $\R_0$ and consider the local asymptotic stability of the disease-free equilibrium $\bX^0$ of \eqref{sys:SEIARS}, we apply the next-generation matrix method of \cite{vdDWatmough2002}.

Define the infected state vector as $(\bE, \bY)$, explicitly including the exposed compartments $\bE$ among infected states.
The vector $\mathcal{F}$ of new infection rates accounts for the incidence of new cases entering the exposed compartments $\bE$ via interactions between susceptible and infectious individuals. The vector $\mathcal{V}$ of all other transitions captures the progression out of $\bE$ (either back to $S$ or onward to $I_1$ and $A_1$) alongside the transitions through the sequential infectious stages. Evaluated at $(\bE,\bY)$, they are given in block vector form by
\begin{equation*}
\mathcal{F} = \begin{pmatrix}
\frac{S}{N} \diag(\bc) \bY \\
\bzero_{2n}
\end{pmatrix}
\quad
\text{and}
\quad
\mathcal{V} = \begin{pmatrix}
(\diag(\bvarepsilon) + d\mathbb{I}_{2n})\bE \\
-\langle \bp_I, \bE \rangle + (\gamma_I + \mu_1 + d)I_1 - \tau_1 A_1 \\
-\gamma_I I_1 + (\gamma_I + \mu_2 + d)I_2 - \tau_2 A_2 \\
\vdots \\
-\gamma_I I_{n-1} + (\gamma_I + \mu_n + d)I_n - \tau_n A_n \\
-\langle \bp_A, \bE \rangle + (\gamma_A + \tau_1 + d)A_1 \\
-\gamma_A A_1 + (\gamma_A + \tau_2 + d)A_2 \\
\vdots \\
-\gamma_A A_{n-1} + (\gamma_A + \tau_n + d)A_n 
\end{pmatrix},
\end{equation*}
where $\diag(\bc)$ is the matrix of contact rates, $\diag(\bvarepsilon) = \diag(\varepsilon_1^I, \ldots, \varepsilon_n^I, \varepsilon_1^A, \ldots, \varepsilon_n^A)$ is the matrix of rates of leaving the exposed compartments and the progression weighting vectors to symptomatic and asymptomatic compartments are defined as $\bp_I = \bm{\pi} \circ \bm{\delta} \circ \bvarepsilon$ and $\bp_A = (\bone_{2n}-\bm{\pi}) \circ \bm{\delta} \circ \bvarepsilon$, respectively.

The Jacobian matrices $F$ and $V$ with respect to $(\bE, \bY)$, evaluated at the DFE, are given by the block matrices
\begin{equation*}
F = \begin{pmatrix}
\bzero_{2n\times 2n} & \diag(\bc) \\
\bzero_{2n\times 2n} & \bzero_{2n\times 2n}
\end{pmatrix}
\quad
\text{and}
\quad 
V = \begin{pmatrix}
\diag(\bvarepsilon) + d\mathbb{I}_{2n} & \bzero_{2n\times 2n} \\
- P & M
\end{pmatrix},
\end{equation*}
where the $2n\times 2n$ matrix $P$ captures the progression from $\bE$ into $I_1$ and $A_1$,
\[
P = \begin{pmatrix} 
\bp_I^T \\ 
\bzero_{n-1\times 2n} \\ 
\bp_A^T \\ 
\bzero_{n-1\times 2n} 
\end{pmatrix},
\]
\[
M = \begin{pmatrix} W_I & -\diag(\btau) \\ \bzero_{n\times n} & W_A \end{pmatrix},
\]
where $\btau = (\tau_1, \ldots, \tau_n)$ and for $X\in\{I,A\}$, $W_X$ takes the form
\[
W_X= \begin{pmatrix}
\alpha_1^X & 0 & \dots & \dots & 0 \\
-\gamma_X & \ddots & \ddots & & \vdots \\
0 & \ddots & \ddots & \ddots & \vdots \\
\vdots & \ddots & \ddots & \ddots & 0 \\
0 & \ldots & 0 & -\gamma_X & \alpha_n^X
\end{pmatrix},
\]
with $\alpha_k^I=\gamma_I+\mu_k+d$ and $\alpha_k^A=\gamma_A+\tau_k+d$.
The basic reproduction number $\R_0$ for \eqref{sys:SEIARS} is defined as the spectral radius of the next-generation matrix $FV^{-1}$.
The inverse of $V$ is computed using block matrix inversion,
\begin{equation*}
V^{-1} = \begin{pmatrix}
(\diag(\bvarepsilon) + d\mathbb{I}_{2n})^{-1} & \bzero_{2n\times 2n} \\
M^{-1} P (\diag(\bvarepsilon) + d\mathbb{I}_{2n})^{-1} & M^{-1}
\end{pmatrix}.
\end{equation*}
Computing the product $FV^{-1}$, we get
\begin{align*}
FV^{-1} &= 
\begin{pmatrix}
\bzero_{2n\times 2n} & \diag(\bc) \\
\bzero_{2n\times 2n} & \bzero_{2n\times 2n}
\end{pmatrix}
\begin{pmatrix}
(\diag(\bvarepsilon) + d\mathbb{I}_{2n})^{-1} & \bzero_{2n\times 2n} \\
M^{-1} P (\bvarepsilon + d\mathbb{I}_{2n})^{-1} & M^{-1}
\end{pmatrix} \\
&=
\begin{pmatrix}
    \diag(\bc) M^{-1} P (\diag(\bvarepsilon) + d\mathbb{I}_{2n})^{-1} & \diag(\bc) M^{-1} \\
    \bzero_{2n\times 2n} & \bzero_{2n\times 2n}
\end{pmatrix},
\end{align*}
so the non-zero eigenvalues of $FV^{-1}$ are entirely determined by the top-left block.
Consequently, the spectral radius $\rho(FV^{-1})$ is the spectral radius of this sub-matrix, giving
\begin{align}
\R_0 &= 
\rho \left( \diag(\bc) 
\begin{pmatrix}
W_I^{-1} & W_I^{-1}\diag(\btau) W_A^{-1} \\
\bzero_{n\times n} & W_A^{-1}\end{pmatrix}
P (\diag(\bvarepsilon) + d\mathbb{I}_{2n})^{-1} \right),
\label{eq:R_0-computation}
\end{align}
which is the expression \eqref{eq:R_0} in the result.
In that matrix, the inverses of $W_X$ take the form
\[
(W_X^{-1})_{ij}=\begin{cases}
\dfrac{1}{\alpha_i^X} & \text{if } i=j \\
\dfrac{\gamma_X^{i-j}}{\prod_{k=j}^i \alpha_k^X} & \text{if } i>j \\
0 & \text{if } i<j.
\end{cases}
\]

Hypotheses (A1) to (A4) in Theorem 2 of \cite{vdDWatmough2002} are automatically satisfied given the method of construction of matrices $F$ and $V$.
So there remains to show that the disease-free equilibrium of the system without disease is always locally asymptotically stable (hypothesis (A5) in \cite[Theorem 2]{vdDWatmough2002}).
In the absence of disease, \eqref{sys:SEIARS} reduces to the $(S,\bE,R)$ subsystem,
\begin{subequations}
\label{sys:population-level-no-disease}
\begin{align}
\dot S &= b + \langle \bvarepsilon \circ (\bone_{2n} - \bm{\delta}), \bE \rangle +\nu R -dS, \nonumber \\
\dot{\bE} &= -(\bvarepsilon+d\mathbb{I}_{2n})\bE, \nonumber\\
\dot R &= -(\nu+d)R, \nonumber
\end{align}
\end{subequations}
whose Jacobian matrix evaluated at the DFE takes the block-triangular form
\[
J=\begin{pmatrix}
    -d & \bu^T & \nu \\
    \bzero_{2n,1} & -(\bvarepsilon+d\mathbb{I}_{2n}) & \bzero_{2n,1} \\
    0 & \bzero_{1,2n} & -(\nu+d)
\end{pmatrix},
\]
where the vector $\bu = \bvarepsilon \circ (\bone_{2n} - \bm{\delta})$.
Because $J$ is block upper-triangular, its eigenvalues are precisely $-d$, $-(\nu+d)$, and the diagonal entries of the diagonal matrix $-(\bvarepsilon+d\mathbb{I}_{2n})$, which are $-(\varepsilon_i^X+d)$. 
All eigenvalues of $J$ are strictly negative and thus the (linear) system \eqref{sys:population-level-no-disease} converges to the equilibrium $(S,\bE^T,R)=(b/d,\bzero_{2n}^T,0)$, i.e., the disease-free equilibrium.
The result follows.

\subsection{Global asymptotic stability of the DFE when $\R_0\leq 1$}
\label{app:proofs-GAS-DFE}
We investigate the global stability of the disease-free equilibrium using the framework of \cite{KamgangSallet2008}. 
This framework requires writing the system in pseudo-triangular form and finding a constant upper-bound Metzler matrix for the infectious subsystem.

We group the non-infectious compartments $x_1 = (S, R)^T$ and the exposed and infectious compartments $x_2 = (\bE^T, \bY^T)^T$. We write \eqref{sys:SEIARS} as:
\begin{subequations}\label{sys:KS-pseudo-triangular-form}
\begin{align}
\dot{x}_1 &= A_1(x)(x_1-x_1^0) + A_{12}(x)x_2, \nonumber \\
\dot{x}_2 &= A_2(x)x_2, \nonumber
\end{align}
\end{subequations}
where the matrix $A_2(x)$ driving the infectious dynamics is given by the block matrix:
\begin{equation*}
A_2(x) = \begin{pmatrix}
-(\diag(\bvarepsilon) + d\mathbb{I}_{2n}) & \frac{S}{N}\diag(\bc) \\
P & -M
\end{pmatrix}.
\end{equation*}

For the domain $\mathcal{D} = \{x = (x_1,x_2) \in \Omega, x_1 \neq 0\}$, the set is positively invariant (see Theorem \ref{th:positivity}). Furthermore, for all $x \in \mathcal{D}$, the off-diagonal elements of $A_2(x)$ are non-negative, meaning $A_2(x)$ is a Metzler matrix.

Because the non-infectious subsystem $\dot{x}_1 = A_1(x_1, 0)(x_1-x_1^0)$ reduces to the linear equations for $(S, R)$ without disease, it trivially globally asymptotically converges to $x_1^0 = (S^0, 0)^T$. Thus, hypothesis $\mathcal{H}_2$ of \cite{KamgangSallet2008} is satisfied.

To apply Theorem 4.3 in \cite{KamgangSallet2008}, we must find a constant matrix $\overline{A}_2$ such that $A_2(x) \leq \overline{A}_2$ for all $x \in \mathcal{D}$. 

A major advantage of structuring the exposed compartment as the tracking vector $\bE$ is that the progression fractions $\bm{\delta}$ and $\bm{\pi}$ are constant parameters. The only state-dependent non-linearity in $A_2(x)$ is the ratio $S/N$. Since $S \leq N$ for all biologically feasible states, $S/N \leq 1$. Thus, the supremum of $A_2(x)$ is evaluated exactly at the DFE, providing the constant upper-bound matrix
\begin{equation*}
\overline{A}_2 = \begin{pmatrix}
-(\diag(\bvarepsilon) + d\mathbb{I}_{2n}) & \diag(\bc) \\
P & -M
\end{pmatrix} \equiv F - V.
\end{equation*}

Since $\overline{A}_2$ is $F - V$ from the Next-Generation Matrix method, standard results from \cite{vdDWatmough2002} guarantee that the maximum real part of the eigenvalues of $\overline{A}_2$ has the same sign as $\R_0 - 1$. 
Consequently, $\overline{A}_2$ is a stable Metzler matrix ($\rho(\overline{A}_2) \le 0$) if and only if $\R_0 \leq 1$. 
Applying \cite[Theorem 4.3]{KamgangSallet2008}, we have the global asymptotic stability of the DFE when $\R_0\leq 1$.

\subsection{Existence and uniqueness of the endemic equilibrium}
\label{app:existence-uniqueness-EEP}

To prove existence and uniqueness of the endemic equilibrium, we first need the following result, which gives an explicit (albeit complicated) expression for the basic reproduction number.
\begin{lemma}
Let $H = \diag(\bc) M^{-1} P (\bvarepsilon + d \II_{2n})^{-1}$. Then, the spectral radius of $H$ satisfies
\begin{equation*}
\rho(H) = 
\rho\left(\diag(\bc) M^{-1}P (\bvarepsilon + d \II_{2n})^{-1}\right) = \rho\left(P (\bvarepsilon + d \II_{2n})^{-1}\diag(\bc) M^{-1}\right).
\end{equation*}
Furthermore, the basic reproduction number \eqref{eq:R_0} can be written as
\begin{equation}\label{eq:R_0-proof-existence-EEP}
\R_0= \frac{(h_{1,1} + h_{n+1,n+1}) + \sqrt{(h_{1,1} - h_{n+1,n+1})^2 + 4 h_{1,n+1} h_{n+1,1}}}{2},
\end{equation}
where the $h_{i,j}$ are given by \eqref{sys:h_ij}.
\end{lemma}

\begin{proof}
First, recall that $\diag(\bc), \diag(\bvarepsilon), M,P \in \mathbb{R}^{2n \times 2n}$. 
The proof relies on a structural property of matrices: for any two square matrices $A, B \in \mathbb{R}^{m \times m}$, the products $AB$ and $BA$ share non-zero eigenvalues with matching algebraic multiplicities \cite[Theorem 1.3.22]{horn2012matrix}.
Consequently, their spectral radii are equal, i.e., $\rho(AB) = \rho(BA)$.
Using this with $A = \diag(\bc) M^{-1}$ and $B = P (\diag(\bvarepsilon) + d \II_{2n})^{-1}$ implies that
\begin{equation*}
\rho\left(\diag(\bc) M^{-1}P (\diag(\bvarepsilon) + d \II_{2n})^{-1}\right) 
= \rho\left(P (\diag(\bvarepsilon) + d \II_{2n})^{-1}\diag(\bc) M^{-1}\right).
\end{equation*}
Now, let 
\begin{equation*}
    \tilde{H}= P  (\varepsilon + d \II_{2n})^{-1} \diag(\bc) M^{-1}
\end{equation*}
Since $\diag(\bc)$ and $(\diag(\bvarepsilon) + d \II_{2n})^{-1}$ are both diagonal matrices, we have
\begin{equation*}
\diag(\bc) (\diag(\bvarepsilon) + d \II_{2n})^{-1} = (\diag(\bvarepsilon) + d \II_{2n})^{-1} \diag(\bc).
\end{equation*}

Left-multiplication by $P$ in $\tilde{H}$ means that $\tilde{H}$ inherits the row-sparsity of $P$, having only rows $1$ and $n+1$ non-zero. 
This reduces the $2n \times 2n$ spectral radius problem to finding the spectral radius of the following $2 \times 2$ matrix,
\begin{equation*}
H_r= \begin{pmatrix}
h_{1,1} & h_{1,n+1} \\
h_{n+1,1} & h_{n+1,n+1}
\end{pmatrix},
\end{equation*} 
where
\begin{subequations}
    \label{sys:h_ij}
    \begin{align}
    h_{1,1} &= \sum_{k=1}^n \frac{\pi_k^I \delta_k^I \varepsilon_k^I c_k^I}{\varepsilon_k^I + d} \;\frac{\gamma_I^{k-1}}{\prod_{m=1}^k \alpha_m^I} \\
    h_{n+1,1} &= \sum_{k=1}^n \frac{(1 - \pi_k^I) \delta_k^I \varepsilon_k^I c_k^I}{\varepsilon_k^I + d} \;\frac{\gamma_I^{k-1}}{\prod_{m=1}^k \alpha_m^I} \\
    h_{1,n+1} &= \sum_{k=1}^n \left[ \frac{\pi_k^I \delta_k^I \varepsilon_k^I c_k^I}{\varepsilon_k^I + d} \left( \sum_{j=1}^k \tau_j \,\frac{\gamma_I^{k-j}}{\prod_{m=j}^k \alpha_m^I} \, \frac{\gamma_A^{i-1}}{\prod_{m=1}^i \alpha_m^A}\right)  \right] \\
    &\qquad + \sum_{k=1}^n \frac{\pi_k^A \delta_k^A \varepsilon_k^A c_k^A}{\varepsilon_k^A + d} \;\frac{\gamma_A^{k-1}}{\prod_{m=1}^k \alpha_m^A} 
    \nonumber \\
    h_{n+1,n+1} &= \sum_{k=1}^n \left[ \frac{(1 - \pi_k^I) \delta_k^I \varepsilon_k^I c_k^I}{\varepsilon_k^I + d} \left( \sum_{j=1}^k \tau_j \, \frac{\gamma_I^{k-j}}{\prod_{m=j}^k \alpha_m^I} \, \frac{\gamma_A^{i-1}}{\prod_{m=1}^i \alpha_m^A}\right)  \right] \\
    & \qquad + \sum_{k=1}^n \frac{(1 - \pi_k^A) \delta_k^A \varepsilon_k^A c_k^A}{\varepsilon_k^A + d} \;\frac{\gamma_A^{k-1}}{\prod_{m=1}^k \alpha_m^A}. 
    \nonumber
    \end{align}
\end{subequations}
Thus, eigenvalues of $H_r$ are solutions of 
\begin{equation*}
z^2-(h_{1,1} + h_{n+1,n+1})z+h_{1,1} h_{n+1,n+1}-h_{n+1,1} h_{1,n+1}=0.
\end{equation*}
It follows that the basic reproduction number $\R_0$ is given by \eqref{eq:R_0-proof-existence-EEP}.
\end{proof}   

We can now tackle the existence and uniqueness proof.

\subsubsection*{States $I_k^\star$ and $A_k^\star$}
For $k\in \{2,\dots,n\}$, the equilibrium relations yield
\begin{subequations}
\label{sys:EP-conditions-EEP}
\begin{align}
I_k^\star &= \dfrac{\gamma_I}{\alpha^I_k}I_{k-1}^\star+\dfrac{\tau_k\gamma_A}{\alpha^I_k\alpha^A_k}A_{k-1}^\star, \label{sys:EP-conditions-EEP-Ik} \\
A_k^\star &= \dfrac{\gamma_A}{\alpha^A_k}A_{k-1}^\star. \label{sys:EP-conditions-EEP-Ak}
\end{align}    
\end{subequations}
By repeated substitution on \eqref{sys:EP-conditions-EEP-Ak}, the expression for the asymptomatic infectious compartments is written for $k\in\{2,\ldots,n\}$ as
\begin{equation}\label{eq:A_k-fct-A_1-EEP}
A_k^\star = \dfrac{\gamma_A^{k-1}}{\prod_{i=2}^{k}\alpha^A_i}A_{1}^\star.
\end{equation}
Similarly, executing iterations on \eqref{sys:EP-conditions-EEP-Ik} leads to
\begin{equation}
I_k^\star = \left( \prod_{j=2}^k \frac{\gamma_I}{\alpha^I_j} \right) I_1^\star + \sum_{m=2}^k \left( \prod_{j=m+1}^k \frac{\gamma_I}{\alpha^I_j} \right) \frac{\tau_m\gamma_A}{\alpha^I_m \alpha^A_m} A_{m-1}^\star.
\end{equation}
Substituting \eqref{eq:A_k-fct-A_1-EEP} into the summation term yields
\begin{equation*}
I_k^\star = \frac{\gamma_I^{k-1}}{\prod_{j=2}^k \alpha^I_j} I_1^\star + A_1^\star \sum_{m=2}^k \left( \frac{\tau_m \;\gamma_I^{k-m} \;\gamma_A^{m-1}}{\prod_{j=m}^k \alpha^I_j \;\prod_{j=2}^{m} \alpha^A_j}\right).
\end{equation*}
Denoting
\begin{equation}
K_{1k} = \frac{\gamma_A^{k-1}}{\prod_{j=1}^{k} \alpha^A_j}, 
\quad 
K_{2k} = \dfrac{\gamma_I^{k-1}}{\prod_{i=2}^{k}\alpha^I_i} 
\quad \text{and} \quad 
K_{3k} = \sum_{m=2}^k \left( \frac{\tau_m \;\gamma_I^{k-m} \;\gamma_A^{m-1}}{\prod_{j=m}^k \alpha^I_j \;\prod_{j=2}^{m} \alpha^A_j}\right),
\end{equation}
we have, for $k\in \{2,\dots,n\}$,
\begin{equation}\label{eq:I_k-A_k-fct-A_1-EEP}
I_k^\star = K_{2k} I_1^\star + K_{3k} A_1^\star 
\quad \text{and}\quad 
A_k^\star = K_{1k} A_1^\star.
\end{equation}

\subsubsection*{States $A_1^\star$, $I_1^\star$ and $R^\star$}
Inflows into the first stage of clinical manifestation $I_1$ and $A_1$ satisfy
\begin{subequations}
\begin{align}
\label{eqEEA1}
I_1^\star &= \frac{1}{\alpha_1^I}\sum_{k=1}^n \left[\left(\pi_k^I+ \frac{(1-\pi_k^I)\tau_1}{\alpha_1^A}\right) \delta_k^I \varepsilon_k^I E_k^{I\star} + \left(\pi_k^A+ \frac{(1-\pi_k^A)\tau_1}{\alpha_1^A}\right)  \delta_k^A\varepsilon_k^A E_k^{ \star A} \right], \\
A_1^\star &= \frac{1}{\alpha_1^A}\sum_{k=1}^n \left[ (1-\pi_k^I) \delta_k^I \varepsilon_k^I E_k^{ \star I} + (1-\pi_k^A) \delta_k^A\varepsilon_k^A E_k^{A\star} \right].
\label{eqEEI1}
\end{align}
\end{subequations}
Furthermore, the equilibrium of the recovered population $R^\star$ takes the form
\begin{equation}
\label{eq:EEP-R}
R^\star = \frac{\gamma_I I_n^\star + \gamma_A A_n^\star}{\nu+d} = 
\dfrac{\gamma_I K_{2n} I_1^\star+\left(\gamma_A K_{1n}+ \gamma_I K_{3n} \right)A_1^\star}{\nu+d}.
\end{equation}

\subsubsection*{States $E_k^I$ and $E_k^A$}
First, recall that the forces of exposure from $I_1$ and $A_1$ at the equilibrium point are, respectively, $\lambda_1^{I\star} = {c_1^I I_1^{\star}}/{N^{\star}}$ and $\lambda_1^{A\star} = {c_1^A A_1^{\star}}/{N^{\star}}$. 
So the remaining forces of exposure, $k \in \{2,\dots,n\}$, satisfy
\begin{subequations}\label{sys:lambda_k-star}
\begin{align}
\lambda_k^{I\star} &= \frac{c_k^I I_k^{\star}}{N^{\star}} = \frac{c_k^I (K_{2k} I_1^{\star}+K_{3k} A_1^\star)}{N^{\star}} = \frac{c_k^I K_{2k}}{c_1^I}\lambda_1^{I\star} + \frac{c_k^I K_{3k}}{c_1^A}\lambda_1^{A\star}, \\
\lambda_k^{A\star} &= \frac{c_k^A A_k^\star}{N^\star} = \frac{c_k^A K_{1k}}{c_1^A}\lambda_1^{A\star}.
\end{align}    
\end{subequations}
Thus, components of the endemic equilibrium point for the exposure classes are governed for $k=2,\dots,n$ by
\begin{align*}
E_k^{I\star} &= \dfrac{\lambda_k^{I\star} S^\star}{\varepsilon_k^I+d} = 
\frac{c_k^I K_{2k} \lambda_1^{I\star} S^{\star}}{c_1^I(\varepsilon_k^I+d)} 
+ \frac{c_k^I K_{3k} \lambda_1^{A\star} S^{\star}}{c_1^A(\varepsilon_k^I+d)}, \\
E_k^{A\star} &= \dfrac{\lambda_k^{A\star} S^\star}{\varepsilon_k^A+d} = 
\dfrac{c_k^A K_{1k}\lambda_1^{A\star}S^{\star}}{c_1^A(\varepsilon_k^A+d)}.
\end{align*}
However, 
$$E_1^{I\star} = \dfrac{\lambda_1^{I\star} S^\star}{\varepsilon_1^I+d}\quad \text{and}\quad E_1^{A\star}= \dfrac{\lambda_1^{A\star} S^\star}{\varepsilon_1^A+d}.$$
This gives that
\begin{equation}
\label{EkEnd}
E_k^{I\star} = S^\star \left( K_{4k}^I \lambda_1^{I\star} + K_{5k}^I \lambda_1^{A\star} \right)
\quad \text{and} \quad 
E_k^{A\star} = S^\star K_{4k}^A \lambda_1^{A\star},
\end{equation}
where
\begin{align*}
K_{41}^A &= \dfrac{1}{\varepsilon_1^A+d}, \quad K_{41}^I= \dfrac{1}{\varepsilon_1^I+d}, \quad  K_{51}^I = 0, \\
K_{4k}^A &= \dfrac{c_k^A K_{1k}}{c_1^A(\varepsilon_k^A+d)},
\quad K_{4k}^I = \dfrac{c_k^I K_{2k}}{c_1^I(\varepsilon_k^I+d)}, \quad K_{5k}^I = \dfrac{c_k^I K_{3k}}{c_1^A(\varepsilon_k^I+d)} \quad (\forall k \ge 2).
\end{align*}
Substituting $E_k^{I\star}$ and $E_k^{A\star}$ into equations \eqref{eqEEA1} and \eqref{eqEEI1}, we have the following relation
\begin{subequations}
\begin{align}
\label{A1_comp} A_1^\star &= S^\star \left( K_{6} \lambda_1^{I\star} + K_{7} \lambda_1^{A\star} \right), \\
\label{I1_comp} I_1^\star &= S^\star \left( K_{8} \lambda_1^{I\star} + K_{9} \lambda_1^{A\star} \right),\\
\label{EER_star}
R^{\star}&= S^{\star}\left(K_{10} \lambda_1^{I\star} +K_{11}\lambda_1^{A\star}\right),
\end{align}    
\end{subequations}
where
\begin{align*}
K_{6} &= \frac{1}{\alpha_1^A}\sum_{k=1}^n (1-\pi_k^I) \delta_k^I \varepsilon_k^I K_{4k}^I,\\
K_{7} &= \frac{1}{\alpha_1^A}\sum_{k=1}^n \left[ (1-\pi_k^I) \delta_k^I \varepsilon_k^I K_{5k}^I + (1-\pi_k^A) \delta_k^A\varepsilon_k^A K_{4k}^A \right], \\
K_{8} &= \frac{1}{\alpha_1^I}\sum_{k=1}^n \left(\pi_k^I+ \frac{(1-\pi_k^I)\tau_1}{\alpha_1^A}\right) \delta_k^I \varepsilon_k^I K_{4k}^I, \\
K_{9} &= \frac{1}{\alpha_1^I}\sum_{k=1}^n \left[\left(\pi_k^I+ \frac{(1-\pi_k^I)\tau_1}{\alpha_1^A}\right) \delta_k^I \varepsilon_k^I K_{5k}^I + \left(\pi_k^A+ \frac{(1-\pi_k^A)\tau_1}{\alpha_1^A}\right)  \delta_k^A\varepsilon_k^A K_{4k}^A \right],\\
K_{10} &= \dfrac{\left(\gamma_A K_{1n}+ \gamma_I K_{3n} \right)K_{6} + \gamma_I K_{2n} K_{8}}{\nu+d} \\
\intertext{and}
K_{11} &= \dfrac{\left(\gamma_A K_{1n}+ \gamma_I K_{3n} \right)K_{7} + \gamma_I K_{2n} K_{9}}{\nu+d}.
\end{align*}

\subsubsection*{State  $S^\star$}
Using equations \eqref{EkEnd} and \eqref{EER_star}, the susceptible individual at the endemic equilibrium is :
\begin{align}
S^\star =& \frac{b + \sum_{k=1}^n \left[ (1-\delta_k^I)\varepsilon_k^I E_k^{I\star} + (1-\delta_k^A)\varepsilon_k^A E_k^{A\star} \right] +\nu R^\star }{\lambda_{\text{exp}}^\star+d} \nonumber\\
=& \frac{b + S^\star \left[ \left( \sum_{k=1}^n (1-\delta_k^I)\varepsilon_k^I K_{4k}^I + \nu K_{10}\right)\lambda_1^{I\star} + \left[\sum_{k=1}^n \left( (1-\delta_k^I)\varepsilon_k^I K_{5k}^I + (1-\delta_k^A)\varepsilon_k^A K_{4k}^A \right) + \nu K_{11}\right]\lambda_1^{A\star} \right]}{\lambda_{\text{exp}}^\star+d}. \label{Send_exp}
\end{align}

In other hand, the expressions of equations \eqref{A1_comp} and \eqref{I1_comp} give that the total force of exposure $\lambda_{\text{exp}}^\star$ is expressed  as a linear combination of $\lambda_1^{I\star}$ and $\lambda_1^{A\star}$. It is
\begin{equation}
\label{end_expo}
\lambda_{\text{exp}}^\star = \lambda_1^{I\star}+\lambda_1^{A\star}+\sum_{k=2}^n \lambda_k^{I\star} + \sum_{k=2}^n \lambda_k^{A\star}= K_{12}\lambda_1^{I\star} + K_{13} \lambda_1^{A\star},
\end{equation}
where
\begin{equation*}
K_{12}= 1+\sum_{k=2}^n \frac{c_k^I K_{2k}}{c_1^I} \quad \text{and} \quad K_{13} =1+ \sum_{k=2}^n \left( \frac{c_k^I K_{3k}}{c_1^A} + \frac{c_k^A K_{1k}}{c_1^A} \right).
\end{equation*}
Then, substituting \eqref{end_expo} into \eqref{Send_exp} gives that
\begin{equation*}
\label{NS_clean}
S^\star = \frac{b}{d + K_{14} \lambda_1^{I\star} + K_{15} \lambda_1^{A\star}},
\end{equation*}
where
\begin{align*}
K_{14} =& K_{12}- (K_{16} + \nu K_{10}),\quad K_{15} = K_{13} - (K_{17} + \nu K_{11})\\
K_{16} =& \sum_{k=1}^n (1-\delta_k^I)\varepsilon_k^I K_{4k}^I \quad\text{and}\quad 
K_{17} = \sum_{k=1}^n \left[ (1-\delta_k^I)\varepsilon_k^I K_{5k}^I + (1-\delta_k^A)\varepsilon_k^A K_{4k}^A \right].
\end{align*}

\subsubsection*{The total population $N^\star$}
The total population equation at the endemic equilibrium is given by
\begin{equation*}
\label{eq:N_definition}
N^\star = S^\star + \sum_{k=1}^n E_k^{I\star} + \sum_{k=1}^n E_k^{A\star} + \sum_{k=1}^n I_k^\star + \sum_{k=1}^n A_k^\star + R^\star.
\end{equation*}
However, we have the following, in which we introduce notations $K_{18}$ to $K_{23}$.
\begin{align*}
\sum_{k=1}^n E_k^{A\star} &= S^\star \left( \sum_{k=1}^n K_{4k}^A \right) \lambda_1^{A\star} =: S^\star K_{18}\lambda_1^{A\star}, \\
\sum_{k=1}^n E_k^{I\star} &= S^\star \left( \sum_{k=1}^n K_{4k}^I \right) \lambda_1^{I\star} + S^\star \left( \sum_{k=1}^n K_{5k}^I \right) \lambda_1^{A\star} =: S^\star K_{19}\lambda_1^{I\star} + S^\star K_{20} \lambda_1^{A\star},\\
\sum_{k=1}^n A_k^\star &= A_1^\star + \left( \sum_{k=2}^n K_{1k} \right) A_1^\star 
=: K_{21} A_1^\star =K_{21} S^\star \left( K_{6} \lambda_1^{I\star} + K_{7} \lambda_1^{A\star} \right),\\
\sum_{k=1}^n I_k^\star &= I_1^\star +\left( \sum_{k=2}^n K_{2k} \right) I_1^\star + \left( \sum_{k=2}^n K_{3k} \right) A_1^\star =: K_{22} I_1^\star + K_{23}A_1^\star,\\
&= K_{22} S^\star \left( K_{8} \lambda_1^{I\star} + K_{9} \lambda_1^{A\star} \right) + K_{23}S^\star \left( K_{6} \lambda_1^{I\star} + K_{7} \lambda_1^{A\star} \right),
\end{align*}
where
\begin{equation*}
K_{18}= \sum_{k=1}^n K_{4k}^A,\,\;K_{19}= \sum_{k=1}^n K_{4k}^I,\,\; K_{20} = \sum_{k=1}^n K_{5k}^I,\,\; K_{21} = 1+\sum_{k=2}^n K_{1k},\,\; K_{22} = 1+\sum_{k=2}^n K_{2k}\,\; \text{and}\,\;K_{23}= \sum_{k=2}^n K_{3k}.
\end{equation*}
Using the above expressions and the fact that $R^\star = S^\star \left( K_{10} \lambda_1^{I\star} + K_{11} \lambda_1^{A\star} \right)$, we have:
\begin{equation}
\label{N_end}
N^\star = S^\star \left[ 1 + K_{24} \lambda_1^{I\star} + K_{25} \lambda_1^{A\star} \right]
\end{equation}
with
\begin{equation*}
K_{24} = K_{19}+ K_{22} K_{8} + (K_{21} + K_{23}) K_{6} + K_{10}\quad \text{and}\quad 
K_{25} = K_{18}+ K_{20} + K_{22} K_{9} + (K_{21} + K_{23}) K_{7} + K_{11}.
\end{equation*}

\subsubsection{Condition for the existence of the endemic equilibrium}
The definitions \eqref{sys:lambda_k-star} give that
\begin{equation}
\label{AI1_N}
A_1^\star = \frac{\lambda_1^{A\star} N^\star}{c_1^A} \quad \text{and} \quad I_1^\star = \frac{\lambda_1^{I\star} N^\star}{c_1^I}.
\end{equation}
Using \eqref{AI1_N} in equations \eqref{A1_comp}-\eqref{I1_comp} gives
\begin{align*}
\frac{\lambda_1^{A\star} N^\star}{c_1^A} &= S^\star \left( K_{6} \lambda_1^{I\star} + K_{7} \lambda_1^{A\star} \right), \\ 
\frac{\lambda_1^{I\star} N^\star}{c_1^I} &= S^\star \left( K_{8} \lambda_1^{I\star} + K_{9} \lambda_1^{A\star} \right). 
\end{align*}

To establish the existence of the endemic equilibrium, we compute ${I_1^\star}/{A_1^\star}$ using \eqref{A1_comp}-\eqref{I1_comp} and $\lambda_1^{I\star}={c_1^I I_1^{\star}}/{N^{\star}} \quad \text{and}\quad \lambda_1^{A\star}={c_1^A A_1^{\star}}/{N^{\star}}$.
It gives:
\begin{equation}
\label{EqlambdaIA}
\frac{c_1^A \lambda_1^{I\star}}{c_1^I \lambda_1^{A\star}} = \frac{K_{8} \lambda_1^{I\star} + K_{9} \lambda_1^{A\star}}{K_{6} \lambda_1^{I\star} + K_{7} \lambda_1^{A\star}}.
\end{equation}
Let $x^\star = {\lambda_1^{I\star}}/{\lambda_1^{A\star}}$. Then, equation \eqref{EqlambdaIA} gives that $x^\star$ is the solution of the following equation.
\begin{equation}
\label{ratio_sol}
\left(c_1^A K_{6}\right) x^2 + \left(c_1^A K_{7} - c_1^I K_{8}\right) x - c_1^I K_{9} = 0.
\end{equation}
It is easy to see that \eqref{ratio_sol} has a unique positive solution which is
\begin{equation}
\label{sol_xstar}
x^\star = \dfrac{-\left(c_1^A K_{7} - c_1^I K_{8}\right) + \sqrt{\left(c_1^A K_{7} - c_1^I K_{8}\right)^2 + 4 c_1^I c_1^A K_{6} K_{9}}}{2 c_1^A K_{6}}.
\end{equation}

Equalizing $A_1^\star$ in \eqref{A1_comp} and that from \eqref{AI1_N}, we obtain the following:
\begin{equation}
\label{Fact_AS}
\frac{\lambda_1^{A\star} N^\star}{c_1^A} = S^\star \left( K_{6} \lambda_1^{I\star} + K_{7} \lambda_1^{A\star} \right),
\end{equation}
Then, using the expression of $N^\star$ in equation \eqref{N_end} gives
\begin{equation*}
\lambda_1^{A\star}S^\star \left[ 1 + \left(K_{24} \frac{\lambda_1^{I\star}}{\lambda_1^{A\star}} + K_{25}  \right)\lambda_1^{A\star}\right] = c_1^A S^\star \left( K_{6} \frac{\lambda_1^{I\star}}{\lambda_1^{A\star}} + K_{7}  \right)\lambda_1^{A\star}.
\end{equation*}
This implies
\begin{equation}
1 +  \left(K_{24} x^{\star} + K_{25}\right)\lambda_1^{A\star} = B^\star,
\end{equation}
where
\begin{equation}
B^\star = c_1^A (K_{6} x^\star + K_{7}).
\end{equation}
The resolution gives that
\begin{equation}
\label{End_lambda_A1}
\lambda_1^{A\star} = \frac{B^\star - 1}{K_{24} x^\star + K_{25}}.
\end{equation}
Since $x^\star>0$, the denominator $K_{24} x^\star + K_{25}>0$.

The final part of the proof consists in showing that the explicit solution $\lambda_1^{A\star}$ given in \eqref{End_lambda_A1} is positive when $\mathcal{R}_0 > 1$. For this, we show that $B^\star > 1 \iff \mathcal{R}_0 > 1$\\
Since $x^\star > 0$, the equation \eqref{ratio_sol}  gives 
\begin{equation}
c_1^A K_{6} x^\star + c_1^A K_{7} - c_1^I K_{8} - \frac{c_1^I K_{9}}{x^\star} = 0
\end{equation}
Recalling $B^\star = c_1^A (K_{6} x^\star + K_{7})$, we substitute it into the above equation to obtain the following identities:
\begin{subequations}\label{sys:proof-GAS-condition-c12}
\begin{align}
\label{M_id1}
c_1^A K_{6} x^\star &= B^\star - c_1^A K_{7} \\
\label{M_id2}
\frac{c_1^I K_{9}}{x^\star} &= B^\star - c_1^I K_{8}.
\end{align}    
\end{subequations}
Multiplying together the left-hand and right-hand sides of \eqref{sys:proof-GAS-condition-c12} cancels out the ratio $x^\star$:
\begin{equation*}
c_1^A K_{6} x^\star\;\frac{c_1^I K_{9}}{x^\star} = (B^\star - c_1^A K_{7})(B^\star - c_1^I K_{8}).
\end{equation*}
Its expanded form is
\begin{equation*}
(B^\star)^2 - (c_1^I K_{8} + c_1^A K_{7})B^\star + c_1^I c_1^A \left(K_{7} K_{8} -  K_{6} K_{9}\right) = 0.
\end{equation*}
Elsewhere, one can observe that 
$$h_{1,1} + h_{n+1,n+1} = c_1^I K_{8}+ c_1^A K_{7} \quad \text{and}\quad h_{1,1} h_{n+1,n+1}-h_{n+1,1} h_{1,n+1}= c_1^I c_1^A \left(K_{7} K_{8} - K_{6} K_{9}\right).$$

Since $B^\star>0$, one has $\mathcal{R}_0=B^\star$. Then
$$ \mathcal{R}_0 > 1 \iff B^\star > 1.$$
Consequently, the condition $\mathcal{R}_0 > 1$ implies that $\lambda_1^{A\star}>0$. This concludes the proof.
\section{The basic SLIARS model}
\label{app:SLIARS-model}
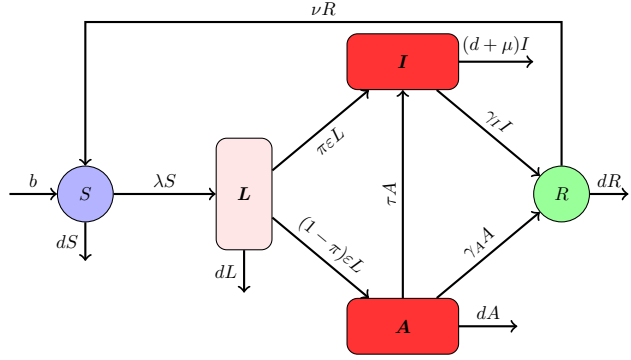
\begin{figure}[htbp]
\centering
\def\hhskip{*3} 
\def\vvskip{*2.5} 
\begin{tikzpicture}[scale=0.7, 
        every node/.style={transform shape},
        auto,
        param/.style={draw, fill=white, rectangle, inner sep=2pt, align=center},
        cloud/.style={minimum width=30pt, draw, circle, fill=gray!20},
        erlang/.style={minimum width=60pt, minimum height=30pt, draw, rounded corners},
        distrib/.style={minimum width=30pt, minimum height=60pt, draw, rounded corners},
        line/.style={draw, -latex'}]    
    \node[cloud, fill=blue!30] (S) at (0\hhskip,0\vvskip) {$S$};
    \node[distrib, fill=red!10] (L) at (1\hhskip,0\vvskip) {$\bL$};
    \node[erlang, fill=red!80] (I) at (2\hhskip,1\vvskip){$\bI$};
    \node[erlang, fill=red!80] (A) at (2\hhskip,-1\vvskip){$\bA$};
    \node[cloud, fill=green!40] (R) at (3\hhskip,0\vvskip){$R$};
    \node [cloud, left=0.3\hhskip of S, draw=none, fill=none] (birthS) {};
    \node [cloud, shift={(270:0.6\hhskip)}, draw=none, fill=none] (deathS) at (S) {};
    \node [cloud, shift={(270:0.8\hhskip)}, draw=none, fill=none] (deathL) at (L) {};
    \node [cloud, shift={(0:1\hhskip)}, draw=none, fill=none] (deathI) at (I) {};
    \node [cloud, shift={(0:0.6\hhskip)}, draw=none, fill=none] (deathR) at (R) {};
    \node [cloud, shift={(0:0.9\hhskip)}, draw=none, fill=none] (deathA) at (A) {};
    \draw[->, thick] (birthS) to node[above] {$b$} (S);
    \draw[->, thick] (S) to node[midway,left] {$dS$} (deathS);
    \draw[->, thick] (L) to node[midway,left] {$dL$} (deathL);
    \draw[->, thick] (I) to node[above,sloped] {$(d+\mu)I$} (deathI);
    \draw[->, thick] (A) to node[above,sloped] {$dA$} (deathA);
    \draw[->, thick] (R) to node[above] {$dR$} (deathR);
    \draw[->, thick] (S) to node [above,sloped] {$\lambda S$} (L);
    \draw[->, thick] (L) to node [below,sloped] {$\pi\varepsilon L$} (I);
    \draw[->, thick] (L) to node [above,sloped,midway] {$(1-\pi)\varepsilon L$} (A);
    \draw[->, thick] (I) to node [above,sloped] {$\gamma_II$} (R);
    \draw[->, thick] (A) to node [above,sloped] {$\gamma_AA$} (R);
    \draw[->, thick] (A) to node [sloped,above] {$\tau A$} (I);
    \draw[->, thick] (R) to (3\hhskip,1.3\vvskip) to node [above] {$\nu R$} (0\hhskip,1.3\vvskip) to (S);
\end{tikzpicture}
\caption{Flowchart of a classic SLIARS model. Horizontal rectangular nodes are ``Erlang segments'', i.e., they comprise a number of similar compartments; see Figure~\ref{fig:flow-within-infectious} for details. The vertical rectangular node is the distribution $\bL$.}
\label{fig:flow-diagram-SLIAR}
\end{figure}

In order to compare comparable models, we retain the Erlang structure of $\bI$ and $\bA$ and correspondingly subdivide the latent compartment into a tracking vector $\mathbf{L} = (L_1^I, \dots, L_n^I, L_1^A, \dots, L_n^A)$ to explicitly record the source of the infection, exactly as we do with the exposed compartments in the exposure model \eqref{sys:SEIARS}.
The flow diagram is shown in Figure~\ref{fig:flow-diagram-SLIAR}.

Unlike the SEIARS model \eqref{sys:SEIARS}, where exposure does not guarantee infection, entering a latent compartment in a classic SLIARS model implies an infection has been successfully acquired, where as frequently, \emph{success} is from the perspective of the pathogen. 
Thus, the effective force of infection incorporates the progression probability $\delta_k^X$ directly at the moment of contact, taking the form
\begin{equation}
\label{eq:force-of-infection-SLIARS}
\lambda_{\text{inf}} = \frac{1}{N} \langle \bm{\delta} \circ \bc, \bY \rangle = \frac{1}{N} \left( \sum_{k=1}^n \delta_k^I c_k^I I_k + \sum_{k=1}^n \delta_k^A c_k^A A_k \right).
\end{equation}
The flow of newly infected individuals into the specific latent compartment $L_k^X$ is $\delta_k^X c_k^X X_kS/N$. Once the latent period concludes, all individuals progress into infectious stages (none return to $S$). 
The model takes the form:
\begin{subequations}
\label{sys:SLIARS}
\begin{align}
\dot S &= b+\nu R -\left(\lambda_{\text{inf}}+d\right)S, \label{sys:SLIARS-dS} \\
\dot L_k^I &= \delta_k^I \frac{c_k^I I_k}{N} S - (\varepsilon_k^I+d)L_k^I, && k=1,\dots,n, \label{sys:SLIARS-dLI} \\
\dot L_k^A &= \delta_k^A \frac{c_k^A A_k}{N} S - (\varepsilon_k^A+d)L_k^A, && k=1,\dots,n, \label{sys:SLIARS-dLA} \\
\dot I_1 &= \sum_{k=1}^n \left[ \pi_k^I \varepsilon_k^I L_k^I + \pi_k^A \varepsilon_k^A L_k^A \right] + \tau_1 A_1 - (\gamma_I+\mu_1+d)I_1, \label{sys:SLIARS-dI1} \\
\dot I_k &= \gamma_I I_{k-1} + \tau_k A_k - (\gamma_I+\mu_k+d)I_k, && k=2,\dots,n, \label{sys:SLIARS-dIk} \\
\dot A_1 &= \sum_{k=1}^n \left[ (1-\pi_k^I) \varepsilon_k^I L_k^I + (1-\pi_k^A) \varepsilon_k^A L_k^A \right] - (\gamma_A+\tau_1+d)A_1, \label{sys:SLIARS-dA1} \\
\dot A_k &= \gamma_A A_{k-1} - (\gamma_A+\tau_k+d)A_k, && k=2,\dots,n, \label{sys:SLIARS-dAk} \\
\dot R &= \gamma_I I_n + \gamma_A A_n - (\nu+d)R. \label{sys:SLIARS-dR}
\end{align}
\end{subequations}
The model is considered with the initial condition $\bX(0)\geq\bzero$.

Applying the next-generation matrix method \cite{vdDWatmough2002}, we determine the basic reproduction number $\R_0$ at the disease-free equilibrium $\bX^0 = (S^0, \bzero)$ where $S^0 = b/d$. 
Defining the infected state vector as $(\mathbf{L}, \bI,\bA)$, the vectors of new infection rates $\mathcal{F}$ and other transitions $\mathcal{V}$ evaluated at the disease-free equilibrium yield the Jacobian matrices $F$ and $V$,
\begin{equation*}
F = \begin{pmatrix}
\bzero_{2n\times 2n} & \diag(\bm{\delta} \circ \bc) \\
\bzero_{2n\times 2n} & \bzero_{2n\times 2n}
\end{pmatrix}, 
\quad 
V = \begin{pmatrix}
\diag(\bvarepsilon) + d\mathbb{I}_{2n} & \bzero_{2n\times 2n} \\
- \tilde{P} & M
\end{pmatrix},
\end{equation*}
where $\diag(\bc)$, $\diag(\bvarepsilon)$ and $M$ are defined identically to the general SEIARS model. 
The $2n \times 2n$ matrix $\tilde{P}$ captures the progression from $\mathbf{L}$ into $I_1$ and $A_1$. 
Because progression to the infectious state is guaranteed upon exiting $\mathbf{L}$, the probability weights strictly rely on the symptom proportions $\bm{\pi}$. 
Defining the vectors $\tilde{\bp}_I = \bm{\pi} \circ \bvarepsilon$ and $\tilde{\bp}_A = (\bone_{2n}-\bm{\pi}) \circ \bvarepsilon$, it takes the form:
\[
\tilde{P} = \begin{pmatrix} 
\tilde{\bp}_I^T \\
\bzero_{n-1\times 2n} \\ 
\tilde{\bp}_A^T \\ 
\bzero_{n-1\times 2n} 
\end{pmatrix}.
\]
By block matrix inversion of $V$, we obtain:
\begin{equation}
\label{eq:R0-SLIARS}
\R_0^{\eqref{sys:SLIARS}} = \rho(FV^{-1}) = 
\rho \left( \diag(\bm{\delta} \circ \bc) M^{-1} \tilde{P} (\diag(\bvarepsilon) + d\mathbb{I}_{2n})^{-1} \right).
\end{equation}

The reproduction numbers of the SEIARS and SLIARS models are actually equal. 
We can rewrite the leading matrix as $\diag(\bm{\delta} \circ \bc) = \diag(\bm{\delta}) \diag(\bc)$. 
Let us compare this with the SEIARS spectral radius from \eqref{eq:R_0-computation},
\begin{align*}
\R_0^{\eqref{sys:SEIARS}} &= \rho \left( \diag(\bc) M^{-1} P (\diag(\bvarepsilon) + d\mathbb{I}_{2n})^{-1} \right).
\end{align*}
Observing that $P = \tilde{P} \diag(\bm{\delta})$ and since diagonal matrices commute, allowing us to write
\[
\diag(\bm{\delta}) (\diag(\bvarepsilon) + d\mathbb{I}_{2n})^{-1} = (\diag(\bvarepsilon) + d\mathbb{I}_{2n})^{-1} \diag(\bm{\delta}),
\]
we have
\begin{align*}
\R_0^{\eqref{sys:SEIARS}} &= \rho \left( \diag(\bc) M^{-1} \tilde{P} (\diag(\bvarepsilon) + d\mathbb{I}_{2n})^{-1} \diag(\bm{\delta}) \right).
\end{align*}
Applying the cyclic property of the spectrum, $\rho(AB) = \rho(BA)$, we shift $\diag(\bm{\delta})$ to the front:
\begin{align*}
\R_0^{\eqref{sys:SEIARS}} &= \rho \left( \diag(\bm{\delta}) \diag(\bc) M^{-1} \tilde{P} (\bvarepsilon + d\mathbb{I}_{2n})^{-1} \right) \\
&= \rho \left( \diag(\bm{\delta} \circ \bc) M^{-1} \tilde{P} (\bvarepsilon + d\mathbb{I}_{2n})^{-1} \right) \equiv \R_0^{\eqref{sys:SLIARS}}.
\end{align*}
Thisguarantees numerical comparisons are generated from the exact same epidemic baseline.

\bibliographystyle{plain}
\bibliography{references,biblio_Arino_Julien}
\end{document}